\documentclass[10pt,letterpaper]{article}

\usepackage[utf8]{inputenc}
\usepackage[T1]{fontenc}
\usepackage[USenglish]{babel}
\usepackage[intlimits]{amsmath}
\usepackage{amsfonts}
\usepackage{amssymb}
\usepackage{amsmath}

\usepackage{mathtools}
\usepackage{setspace}
\usepackage{graphicx}
\usepackage{hyperref}

\usepackage[dvipsnames]{xcolor}
\usepackage{todonotes}
\usepackage{tikz}
\usetikzlibrary{shapes,decorations,shapes.geometric,backgrounds,fit,positioning}
\usepackage{algorithm}
\usepackage{comment}
\usepackage[noend]{algpseudocode}

\usepackage{apxproof}
\theoremstyle{plain}
\newtheoremrep{theorem}{Theorem}
\newtheoremrep{proposition}[theorem]{Proposition}
\newtheoremrep{lemma}[theorem]{Lemma}
\newtheoremrep{claim}[theorem]{Claim}
\newtheoremrep{conjecture}[theorem]{Conjecture}
\newtheoremrep{corollary}[theorem]{Corollary}
\theoremstyle{definition}
\newtheoremrep{definition}[theorem]{Definition}
\theoremstyle{remark}
\newtheoremrep{example}[theorem]{Example}
\newtheoremrep{remark}[theorem]{Remark}

\newcommand{\Oh}{\mathcal{O}}

\usepackage{xspace}

\newcommand\NP{\ensuremath{\textsf{NP}}\xspace}
\newcommand\PSPACE{\ensuremath{\textsf{PSPACE}}\xspace}
\newcommand\Ptime{\ensuremath{\textsf{P}}\xspace}

\usepackage[nocompress]{cite}

\newcommand{\ra}[2]{
\draw [line width=0.5mm, black] (#1 + 0.7,#2 + 0.5) -- (#1 + 0.3,#2 + 0.75);
\draw [line width=0.5mm, black] (#1 + 0.7,#2 + 0.5) -- (#1 + 0.3,#2 + 0.25);
}

\newcommand{\la}[2]{
\draw [line width=0.5mm, black] (#1 + 0.3,#2 + 0.5) -- (#1 + 0.7,#2 + 0.75);
\draw [line width=0.5mm, black] (#1 + 0.3,#2 + 0.5) -- (#1 + 0.7,#2 + 0.25);
}

\newcommand{\ua}[2]{
\draw [line width=0.5mm, black] (#1 + 0.5,#2 + 0.7) -- (#1 + 0.75,#2 + 0.3);
\draw [line width=0.5mm, black] (#1 + 0.5,#2 + 0.7) -- (#1 + 0.25,#2 + 0.3);
}

\newcommand{\da}[2]{
\draw [line width=0.5mm, black] (#1 + 0.5,#2 + 0.3) -- (#1 + 0.75,#2 + 0.7);
\draw [line width=0.5mm, black] (#1 + 0.5,#2 + 0.3) -- (#1 + 0.25,#2 + 0.7);
}

\newcommand{\gra}[2]{
\draw [line width=0.5mm, gray] (#1 + 0.7,#2 + 0.5) -- (#1 + 0.3,#2 + 0.75);
\draw [line width=0.5mm, gray] (#1 + 0.7,#2 + 0.5) -- (#1 + 0.3,#2 + 0.25);
}

\newcommand{\cra}[2]{
	\draw [line width=0.5mm, cyan] (#1 + 0.7,#2 + 0.5) -- (#1 + 0.3,#2 + 0.75);
	\draw [line width=0.5mm, cyan] (#1 + 0.7,#2 + 0.5) -- (#1 + 0.3,#2 + 0.25);
}

\newcommand{\mra}[2]{
\draw [line width=0.5mm, magenta] (#1 + 0.7,#2 + 0.5) -- (#1 + 0.3,#2 + 0.75);
\draw [line width=0.5mm, magenta] (#1 + 0.7,#2 + 0.5) -- (#1 + 0.3,#2 + 0.25);
}

\newcommand{\cla}[2]{
\draw [line width=0.5mm, cyan] (#1 + 0.3,#2 + 0.5) -- (#1 + 0.7,#2 + 0.75);
\draw [line width=0.5mm, cyan] (#1 + 0.3,#2 + 0.5) -- (#1 + 0.7,#2 + 0.25);
}

\newcommand{\gla}[2]{
	\draw [line width=0.5mm, gray] (#1 + 0.3,#2 + 0.5) -- (#1 + 0.7,#2 + 0.75);
	\draw [line width=0.5mm, gray] (#1 + 0.3,#2 + 0.5) -- (#1 + 0.7,#2 + 0.25);
}

\newcommand{\bla}[2]{
\draw [line width=0.5mm, blue] (#1 + 0.3,#2 + 0.5) -- (#1 + 0.7,#2 + 0.75);
\draw [line width=0.5mm, blue] (#1 + 0.3,#2 + 0.5) -- (#1 + 0.7,#2 + 0.25);
}

\newcommand{\mla}[2]{
\draw [line width=0.5mm, magenta] (#1 + 0.3,#2 + 0.5) -- (#1 + 0.7,#2 + 0.75);
\draw [line width=0.5mm, magenta] (#1 + 0.3,#2 + 0.5) -- (#1 + 0.7,#2 + 0.25);
}

\newcommand{\gua}[2]{
\draw [line width=0.5mm, gray] (#1 + 0.5,#2 + 0.7) -- (#1 + 0.75,#2 + 0.3);
\draw [line width=0.5mm, gray] (#1 + 0.5,#2 + 0.7) -- (#1 + 0.25,#2 + 0.3);
}

\newcommand{\cua}[2]{
	\draw [line width=0.5mm, cyan] (#1 + 0.5,#2 + 0.7) -- (#1 + 0.75,#2 + 0.3);
	\draw [line width=0.5mm, cyan] (#1 + 0.5,#2 + 0.7) -- (#1 + 0.25,#2 + 0.3);
}

\newcommand{\bua}[2]{
\draw [line width=0.5mm, blue] (#1 + 0.5,#2 + 0.7) -- (#1 + 0.75,#2 + 0.3);
\draw [line width=0.5mm, blue] (#1 + 0.5,#2 + 0.7) -- (#1 + 0.25,#2 + 0.3);
}

\newcommand{\mua}[2]{
\draw [line width=0.5mm, magenta] (#1 + 0.5,#2 + 0.7) -- (#1 + 0.75,#2 + 0.3);
\draw [line width=0.5mm, magenta] (#1 + 0.5,#2 + 0.7) -- (#1 + 0.25,#2 + 0.3);
}

\newcommand{\gda}[2]{
\draw [line width=0.5mm, gray] (#1 + 0.5,#2 + 0.3) -- (#1 + 0.75,#2 + 0.7);
\draw [line width=0.5mm, gray] (#1 + 0.5,#2 + 0.3) -- (#1 + 0.25,#2 + 0.7);
}

\newcommand{\cda}[2]{
	\draw [line width=0.5mm, cyan] (#1 + 0.5,#2 + 0.3) -- (#1 + 0.75,#2 + 0.7);
	\draw [line width=0.5mm, cyan] (#1 + 0.5,#2 + 0.3) -- (#1 + 0.25,#2 + 0.7);
}

\newcommand{\bda}[2]{
\draw [line width=0.5mm, blue] (#1 + 0.5,#2 + 0.3) -- (#1 + 0.75,#2 + 0.7);
\draw [line width=0.5mm, blue] (#1 + 0.5,#2 + 0.3) -- (#1 + 0.25,#2 + 0.7);
}

\newcommand{\mda}[2]{
\draw [line width=0.5mm, magenta] (#1 + 0.5,#2 + 0.3) -- (#1 + 0.75,#2 + 0.7);
\draw [line width=0.5mm, magenta] (#1 + 0.5,#2 + 0.3) -- (#1 + 0.25,#2 + 0.7);
}

\newcommand{\DA}{\begin{tikzpicture}[scale=0.32]\useasboundingbox (0,0.2) rectangle (1, 1.2);	\draw [step =1, black]grid (1,1);\da{0}{-0.0}\end{tikzpicture}}
\newcommand{\UA}{\begin{tikzpicture}[scale=0.32]\useasboundingbox (0,0.2) rectangle (1, 1.2);\draw [step =1, black]grid (1,1);\ua{0}{0}\end{tikzpicture}}
\newcommand{\RA}{\begin{tikzpicture}[scale=0.32]\useasboundingbox (0,0.2) rectangle (1, 1.2); \draw [step =1, black]grid (1,1);\ra{0}{0}\end{tikzpicture}}
\newcommand{\LA}{\begin{tikzpicture}[scale=0.32]\useasboundingbox (0,0.2) rectangle (1, 1.2);\draw [step =1, black]grid (1,1);\la{0}{0}\end{tikzpicture}}

\newcommand{\LB}{\left[}
\newcommand{\LP}{\left(}
\newcommand{\RB}{\right]}
\newcommand{\RP}{\right)}

\newcommand{\ix}[2]{
\node[text width=1cm, anchor=west] at (#1 - 0.08,#2 + 0.3)
    {$(#1,\!#2)$};
}

\newcommand{\literal}[2]{
\draw [line width=1mm, black] (#1 + 0,#2 + 0) -- (#1 + 4,#2 + 0);
\draw [line width=1mm, black] (#1 + 6,#2 + 0) -- (#1 + 10,#2 + 0);
\draw [line width=1mm, black] (#1 + 4,#2 + 1) -- (#1 + 6,#2 + 1);
\draw [line width=1mm, black] (#1 + 0,#2 + 3) -- (#1 + 3,#2 + 3);
\draw [line width=1mm, black] (#1 + 7,#2 + 3) -- (#1 + 10,#2 + 3);
\draw [line width=1mm, black] (#1 + 3,#2 + 6) -- (#1 + 7,#2 + 6);
\draw [line width=1mm, black] (#1 + 0,#2 + 0) -- (#1 + 0,#2 + 3);
\draw [line width=1mm, black] (#1 + 3,#2 + 3) -- (#1 + 3,#2 + 6);
\draw [line width=1mm, black] (#1 + 4,#2 + 0) -- (#1 + 4,#2 + 1);
\draw [line width=1mm, black] (#1 + 6,#2 + 0) -- (#1 + 6,#2 + 1);
\draw [line width=1mm, black] (#1 + 7,#2 + 3) -- (#1 + 7,#2 + 6);
\draw [line width=1mm, black] (#1 + 10,#2 + 0) -- (#1 + 10,#2 + 3);

\draw [line width=1mm, black] (#1 + 7,#2 + 8) -- (#1 + 9,#2 + 8);

\draw [line width=1mm, black] (#1 + 7,#2 + 7) -- (#1 + 7,#2 + 8);
\draw [line width=1mm, black] (#1 + 7,#2 + 6) -- (#1 + 7,#2 + 7);
\draw [line width=1mm, black] (#1 + 9,#2 + 3) -- (#1 + 9,#2 + 7);
\draw [line width=1mm, black] (#1 + 9,#2 + 7) -- (#1 + 9,#2 + 8);
}

\newcommand{\literaln}[2]{
\draw [line width=1mm, black] (#1 + 0,#2 + 0) -- (#1 + 4,#2 + 0);
\draw [line width=1mm, black] (#1 + 6,#2 + 0) -- (#1 + 10,#2 + 0);
\draw [line width=1mm, black] (#1 + 4,#2 + 1) -- (#1 + 6,#2 + 1);
\draw [line width=1mm, black] (#1 + 0,#2 + 3) -- (#1 + 3,#2 + 3);
\draw [line width=1mm, black] (#1 + 7,#2 + 3) -- (#1 + 10,#2 + 3);
\draw [line width=1mm, black] (#1 + 3,#2 + 6) -- (#1 + 7,#2 + 6);
\draw [line width=1mm, black] (#1 + 0,#2 + 0) -- (#1 + 0,#2 + 3);
\draw [line width=1mm, black] (#1 + 3,#2 + 3) -- (#1 + 3,#2 + 6);
\draw [line width=1mm, black] (#1 + 4,#2 + 0) -- (#1 + 4,#2 + 1);
\draw [line width=1mm, black] (#1 + 6,#2 + 0) -- (#1 + 6,#2 + 1);
\draw [line width=1mm, black] (#1 + 7,#2 + 3) -- (#1 + 7,#2 + 6);
\draw [line width=1mm, black] (#1 + 10,#2 + 0) -- (#1 + 10,#2 + 3);

\draw [line width=1mm, black] (#1 + 2,#2 + 8) -- (#1 + 7,#2 + 8);

\draw [line width=1mm, black] (#1 + 2,#2 + 3) -- (#1 + 2,#2 + 8);
\draw [line width=1mm, black] (#1 + 7,#2 + 6) -- (#1 + 7,#2 + 8);
}

\newcommand{\variable}[2]{
\draw [line width=0.6mm, black] (#1 + 3,#2 + 0) -- (#1 + 7,#2 + 0);
\draw [line width=0.6mm, black] (#1 + 7,#2 + 0) -- (#1 + 7,#2 + 1);
\draw [line width=0.6mm, black] (#1 + 7,#2 + 1) -- (#1 + 9,#2 + 1);
\draw [line width=0.6mm, black] (#1 + 9,#2 + 1) -- (#1 + 9,#2 + 0);
\draw [line width=0.6mm, black] (#1 + 9,#2 + 0) -- (#1 + 13,#2 + 0);
\draw [line width=0.6mm, black] (#1 + 13,#2 + 0) -- (#1 + 13,#2 + 7);
\draw [line width=0.6mm, black] (#1 + 13,#2 + 7) -- (#1 + 11,#2 + 7);
\draw [line width=0.6mm, black] (#1 + 11,#2 + 7) -- (#1 + 11,#2 + 9);
\draw [line width=0.6mm, black] (#1 + 11,#2 + 9) -- (#1 + 10,#2 + 9);
\draw [line width=0.6mm, black] (#1 + 10,#2 + 9) -- (#1 + 10,#2 + 11);
\draw [line width=0.6mm, black] (#1 + 10,#2 + 11) -- (#1 + 6,#2 + 11);
\draw [line width=0.6mm, black] (#1 + 6,#2 + 11) -- (#1 + 6,#2 + 10);
\draw [line width=0.6mm, black] (#1 + 6,#2 + 10) -- (#1 + 4,#2 + 10);
\draw [line width=0.6mm, black] (#1 + 4,#2 + 10) -- (#1 + 4,#2 + 11);
\draw [line width=0.6mm, black] (#1 + 4,#2 + 11) -- (#1 + 0,#2 + 11);
\draw [line width=0.6mm, black] (#1 + 0,#2 + 11) -- (#1 + 0,#2 + 5);
\draw [line width=0.6mm, black] (#1 + 0,#2 + 5) -- (#1 + 3,#2 + 5);
\draw [line width=0.6mm, black] (#1 + 3,#2 + 5) -- (#1 + 3,#2 + 4);
\draw [line width=0.6mm, black] (#1 + 3,#2 + 4) -- (#1 + 2,#2 + 4);
\draw [line width=0.6mm, black] (#1 + 2,#2 + 4) -- (#1 + 2,#2 + 2);
\draw [line width=0.6mm, black] (#1 + 2,#2 + 2) -- (#1 + 3,#2 + 2);
\draw [line width=0.6mm, black] (#1 + 3,#2 + 2) -- (#1 + 3,#2 + 0);
}

\newcommand{\clause}[2]{
\draw [line width=0.6mm, black] (#1 + 0,#2 + 0) -- (#1 + 4,#2 + 0);
\draw [line width=0.6mm, black] (#1 + 4,#2 + 0) -- (#1 + 4,#2 + 1);
\draw [line width=0.6mm, black] (#1 + 4,#2 + 1) -- (#1 + 6,#2 + 1);
\draw [line width=0.6mm, black] (#1 + 6,#2 + 1) -- (#1 + 6,#2 + 0);
\draw [line width=0.6mm, black] (#1 + 6,#2 + 0) -- (#1 + 10,#2 + 0);
\draw [line width=0.6mm, black] (#1 + 10,#2 + 0) -- (#1 + 10,#2 + 3);
\draw [line width=0.6mm, black] (#1 + 10,#2 + 3) -- (#1 + 12,#2 + 3);
\draw [line width=0.6mm, black] (#1 + 12,#2 + 3) -- (#1 + 12,#2 + 0);
\draw [line width=0.6mm, black] (#1 + 12,#2 + 0) -- (#1 + 16,#2 + 0);
\draw [line width=0.6mm, black] (#1 + 16,#2 + 0) -- (#1 + 16,#2 + 1);
\draw [line width=0.6mm, black] (#1 + 16,#2 + 1) -- (#1 + 18,#2 + 1);
\draw [line width=0.6mm, black] (#1 + 18,#2 + 1) -- (#1 + 18,#2 + 0);
\draw [line width=0.6mm, black] (#1 + 18,#2 + 0) -- (#1 + 22,#2 + 0);
\draw [line width=0.6mm, black] (#1 + 22,#2 + 0) -- (#1 + 22,#2 + 3);
\draw [line width=0.6mm, black] (#1 + 22,#2 + 3) -- (#1 + 24,#2 + 3);
\draw [line width=0.6mm, black] (#1 + 24,#2 + 3) -- (#1 + 24,#2 + 0);
\draw [line width=0.6mm, black] (#1 + 24,#2 + 0) -- (#1 + 28,#2 + 0);
\draw [line width=0.6mm, black] (#1 + 28,#2 + 0) -- (#1 + 28,#2 + 1);
\draw [line width=0.6mm, black] (#1 + 28,#2 + 1) -- (#1 + 30,#2 + 1);
\draw [line width=0.6mm, black] (#1 + 30,#2 + 1) -- (#1 + 30,#2 + 0);
\draw [line width=0.6mm, black] (#1 + 30,#2 + 0) -- (#1 + 34,#2 + 0);
\draw [line width=0.6mm, black] (#1 + 34,#2 + 0) -- (#1 + 34,#2 + 9);
\draw [line width=0.6mm, black] (#1 + 34,#2 + 9) -- (#1 + 0,#2 + 9);
\draw [line width=0.6mm, black] (#1 + 0,#2 + 9) -- (#1 + 0,#2 + 0);
}

\newcommand{\clauser}[2]{
\draw [line width=0.6mm, black] (#1 + 0,#2 - 0) -- (#1 + 4,#2 - 0);
\draw [line width=0.6mm, black] (#1 + 4,#2 - 0) -- (#1 + 4,#2 - 1);
\draw [line width=0.6mm, black] (#1 + 4,#2 - 1) -- (#1 + 6,#2 - 1);
\draw [line width=0.6mm, black] (#1 + 6,#2 - 1) -- (#1 + 6,#2 - 0);
\draw [line width=0.6mm, black] (#1 + 6,#2 - 0) -- (#1 + 10,#2 - 0);
\draw [line width=0.6mm, black] (#1 + 10,#2 - 0) -- (#1 + 10,#2 - 3);
\draw [line width=0.6mm, black] (#1 + 10,#2 - 3) -- (#1 + 12,#2 - 3);
\draw [line width=0.6mm, black] (#1 + 12,#2 - 3) -- (#1 + 12,#2 - 0);
\draw [line width=0.6mm, black] (#1 + 12,#2 - 0) -- (#1 + 16,#2 - 0);
\draw [line width=0.6mm, black] (#1 + 16,#2 - 0) -- (#1 + 16,#2 - 1);
\draw [line width=0.6mm, black] (#1 + 16,#2 - 1) -- (#1 + 18,#2 - 1);
\draw [line width=0.6mm, black] (#1 + 18,#2 - 1) -- (#1 + 18,#2 - 0);
\draw [line width=0.6mm, black] (#1 + 18,#2 - 0) -- (#1 + 22,#2 - 0);
\draw [line width=0.6mm, black] (#1 + 22,#2 - 0) -- (#1 + 22,#2 - 3);
\draw [line width=0.6mm, black] (#1 + 22,#2 - 3) -- (#1 + 24,#2 - 3);
\draw [line width=0.6mm, black] (#1 + 24,#2 - 3) -- (#1 + 24,#2 - 0);
\draw [line width=0.6mm, black] (#1 + 24,#2 - 0) -- (#1 + 28,#2 - 0);
\draw [line width=0.6mm, black] (#1 + 28,#2 - 0) -- (#1 + 28,#2 - 1);
\draw [line width=0.6mm, black] (#1 + 28,#2 - 1) -- (#1 + 30,#2 - 1);
\draw [line width=0.6mm, black] (#1 + 30,#2 - 1) -- (#1 + 30,#2 - 0);
\draw [line width=0.6mm, black] (#1 + 30,#2 - 0) -- (#1 + 34,#2 - 0);
\draw [line width=0.6mm, black] (#1 + 34,#2 - 0) -- (#1 + 34,#2 - 9);
\draw [line width=0.6mm, black] (#1 + 34,#2 - 9) -- (#1 + 0,#2 - 9);
\draw [line width=0.6mm, black] (#1 + 0,#2 - 9) -- (#1 + 0,#2 - 0);
}

\newcommand{\edge}[2]{
\draw [line width=0.6mm, black] (#1 + 0,#2 + 0) -- (#1 + 4,#2 + 0);
\draw [line width=0.6mm, black] (#1 + 4,#2 + 0) -- (#1 + 4,#2 + 3);
\draw [line width=0.6mm, black] (#1 + 0,#2 + 3) -- (#1 + 4,#2 + 3);
\draw [line width=0.6mm, black] (#1 + 0,#2 + 3) -- (#1 + 0,#2 + 0);
}

\newcommand{\fanout}[2]{
\draw [line width=0.6mm, black] (#1 + 0,#2 + 0) -- (#1 + 4,#2 + 0);
\draw [line width=0.6mm, black] (#1 + 4,#2 + 1) -- (#1 + 4,#2 + 0);
\draw [line width=0.6mm, black] (#1 + 4,#2 + 1) -- (#1 + 6,#2 + 1);
\draw [line width=0.6mm, black] (#1 + 6,#2 + 0) -- (#1 + 6,#2 + 1);
\draw [line width=0.6mm, black] (#1 + 6,#2 + 0) -- (#1 + 10,#2 + 0);
\draw [line width=0.6mm, black] (#1 + 10,#2 + 3) -- (#1 + 10,#2 + 0);
\draw [line width=0.6mm, black] (#1 + 10,#2 + 3) -- (#1 + 6,#2 + 3);
\draw [line width=0.6mm, black] (#1 + 6,#2 + 2) -- (#1 + 6,#2 + 3);
\draw [line width=0.6mm, black] (#1 + 6,#2 + 2) -- (#1 + 4,#2 + 2);
\draw [line width=0.6mm, black] (#1 + 4,#2 + 3) -- (#1 + 4,#2 + 2);
\draw [line width=0.6mm, black] (#1 + 4,#2 + 3) -- (#1 + 0,#2 + 3);
\draw [line width=0.6mm, black] (#1 + 0,#2 + 0) -- (#1 + 0,#2 + 3);
}

\newcommand{\arrowloopl}[3]{
\foreach \i in {#1,...,#2}{
\la{\i}{#3}
}
}

\makeatletter
\def\BState{\State\hskip-\ALG@thistlm}
\makeatother

\begin{document}
\title{All Paths Lead to Rome} 



\author{Kevin Goergen, 
 Henning Fernau, 
 Esther Oest, 
 Petra Wolf\thanks{Supported by Deutsche Forschungsgemeinschaft (DFG), project
	FE560/9-1.}
}



\date{Universit\"at Trier, Fachbereich~4 -- Abteilung Informatikwissenschaften\\  
54286 Trier, Germany.\\
\{s4kegoer,fernau,s4esmeck,wolfp\}@uni-trier.de}

\maketitle

\begin{abstract}
\emph{All roads lead to Rome} is the core idea of the puzzle game \emph{Roma}.
It is played on an $n \times n$ grid consisting of quadratic cells. Those cells are grouped into boxes of at most four neighboring cells and are either filled, or to be filled, with arrows pointing in cardinal directions. The goal of the game is to fill the empty cells with arrows such that each box contains at most one arrow of each direction and regardless where we start, if we follow the arrows in the cells, we will always end up in the special Roma-cell.
In this work, we study the computational complexity of the puzzle game Roma and show that completing a Roma board according to the rules is an \NP-complete task, counting the number of valid completions is $\#$\Ptime-complete, and determining the number of preset arrows needed to make the instance \emph{uniquely} solvable is $\Sigma_2^P$-complete.
We further show that the problem of completing a given Roma instance on an $n\times n$ board cannot be solved in time $\mathcal{O}\left(2^{{o}(n)}\right)$ under ETH and give a matching dynamic programming algorithm based on the idea of Catalan structures.
\end{abstract}

\section{Introduction}
With computational devices in nearly everyone's pockets nowadays, the opportunities to play puzzle games on these devices are plentiful.
What makes such games so addictive that they are played every day by millions of people?
One possible answer to the suggested question is that (generalized variants of) these games are computationally intractable~\cite{demaine2001playing, kendall2008survey}, which could explain why it can be so challenging to find a solution or to get a good score.
 In this work, we study the puzzle game \emph{Roma} (a playable version can be found in~\cite{Roma}), which we describe in more detail in the next section. Roma has similarities to other puzzle games, such as the famous Sudoku puzzle, shown to be \NP-complete in~\cite{DBLP:journals/ieiceta/YatoS03}, in the sense that the player has to fill out fields in a two-dimensional board, taking into account hints and restrictions given by the concrete instance of the game. 
Puzzle games of this sort such as Kakuro~\cite{DBLP:conf/fun/RueppH10}, Herugolf and Makaro~\cite{DBLP:conf/fun/IwamotoHI18}, Dosun-Fuwari~\cite{DBLP:journals/jip/IwamotoI18}, or Ying-Yang puzzles~\cite{DBLP:conf/cccg/DemaineLRU21} were shown to be \NP-complete.
However, Roma is motivated by a concrete \emph{planning task}: the player is asked to design a map of one-way roads under certain restrictions so that finally, one can travel to the central place (Rome) from each position on the map. Roma can hence be used to explain the difficulties of planning and design in a playful way.

The field of computational complexity of games and computer games is a broad vivid field as it also allows a playful entry to the field of computational complexity theory, see for instance the surveys by Demaine et~al.~\cite{demaine2001playing} and Kendall et~al.~\cite{kendall2008survey}. The importance of the field is also reflected in a huge number of publications at different international conferences over decades, such as in the conferences JCDCG3~\cite{DBLP:journals/gc/AkiyamaISU20} and FUN~\cite{DBLP:conf/fun/2022}.
The study of games is not only a fun topic but also allows for a deeper understanding of fundamental concepts in theoretical computer science. For instance, the game of cops and robbers played on some graph $G$ has a direct connection with the treewidth of $G$~\cite{bonato2011game,seymour1993graph}, 
one of the most important structural parameters in parameterized complexity theory. 
Further, game variants of problems can be used to study a problem from a different perspective, for instance, a two-player variant of the satisfiability problem is equivalent to the \PSPACE-complete quantified SAT problem and can be even harder if we inherit rules native to the game Go~\cite{DBLP:journals/jacm/ChandraKS81,DBLP:journals/siamcomp/StockmeyerC79,DBLP:conf/mfcs/Robson84}.
Hence, studying (computer) games is a great way to better understand combinatorial concepts. 
In recent years, two notable lines of research developed in this field. One is trying to generalize the combinatorial key mechanics of a game and studies the complexity of this combinatorial mechanics through so-called \emph{metatheorems}   \cite{demaine2008constraint,forivsek2010computational,Viglietta2014,demaine2016computational,DBLP:conf/fun/HaanW18}.
Another line of research focuses more on the individual games~\cite{robson1983complexity, robson1984n,DBLP:journals/jct/FraenkelL81,DBLP:journals/tcs/Viglietta15,DBLP:conf/fun/ChurchillBH21,DBLP:conf/fun/BrunnerCDHHSZ21,DBLP:conf/fun/Almanza0P18} 
and also takes a deeper analysis with respect to parameterized complexity theory
\cite{DBLP:conf/concur/BruyereHR18,DBLP:conf/icalp/BonnetGLRS17,bjorklund2003fixed,Feretal03}.
With this work, we are going to contribute to the second line: we analyze the complexity of the game \emph{Roma} also from a parameterized angle.

\paragraph{Our contribution.}
We show that the question whether a partially filled instance of Roma can be completed according to the rules of Roma is an \NP-complete problem by a reduction from \textsc{Planar 3SAT} (Theorem~\ref{thm:NPRoma}). As this reduction is parsimonious, we directly get that the counting variant of Roma, counting the number of solutions, is $\#$\Ptime-complete (Theorem~\ref{thm:CountRoma}). The parsimonious reduction further implies that the question of how many hints must be added to a Roma instance in order to make it \emph{uniquely} solvable is $\Sigma_2^P$ complete (Theorem~\ref{thm:FCPRoma}). We show that the reduction by Lichtenstein~\cite{DBLP:journals/siamcomp/Lichtenstein82} from \textsc{3SAT} to \textsc{Planar 3SAT} can be translated into our Roma setting with only a constant factor increase in space. Especially, we have that the number of variables and the number of clauses each correspond to the dimension $n$ of an $n\times n$ Roma board and hence, assuming ETH, Roma cannot be solved in time $\mathcal{O}\left(2^{{o}(n)}\right)$ (Theorem~\ref{thm:ETH-Roma}).
As our second main result, we match this lower bound by a dynamic programming algorithm, we believe to be interesting, using the idea of Catalan structures (Theorem~\ref{thm:DPalgo}).

\subsection*{The Rules of Roma}
Roma is a one-person puzzle game. A Roma board consists of a quadratic game board, which in turn consists of $n \times n$ quadratic individual cells. One of these cells is a previously determined \emph{Roma-cell} which serves as a target cell. Cells which directly border on each other are called \emph{true neighbors}.\footnote{In  cellular automata theory, this notion of neighborhood is known as \emph{von-Neumann-neighborhood}. In image processing, this resembles the notion of  \emph{4-connected pixels}.}  Cells which are true neighbors can be gathered in a collection called a \emph{box}. These boxes can consist of 1 to 4 cells. The boxes are preset at the beginning of a game and every cell is contained in one box. The boxes can take any form, as long as every cell within a box can be reached from any other cell within that box by only traversing other cells from the same box, where traversal refers to single-cell steps from one cell to one of its true neighbors. The Roma-cell is always contained in its own 1-box. Each empty cell must be filled by the player with an arrow, pointing in one of the four cardinal directions. Cells can contain preset arrows before the game starts. Each box can contain only one arrow pointing in a given cardinal direction.
The goal of the game is to fill each cell in such a manner that, beginning in any cell within the board, following the arrows step by step will always lead to the Roma-cell. 
An example game board may look as displayed in \autoref{fig:RomaEx}. 
A more mathematical description of the game will follow next.\footnote{From here on, we will refer to the formal decision problem as \textsc{Roma} as opposed to the game Roma.}

\begin{figure}[tb]
\begin{center}
\begin{minipage}{.23\textwidth}
\scalebox{.6}{\begin{tikzpicture}[scale=0.79]
\draw [step =1, black] (-1,-1) grid (5,5);
\draw (1.5,2.5) circle (8pt);

\draw [line width=1.1mm, black] (0,0) -- (4,0);
\draw [line width=1.1mm, black] (2,1) -- (3,1);
\draw [line width=1.1mm, black] (1,2) -- (2,2);
\draw [line width=1.1mm, black] (3,2) -- (4,2);
\draw [line width=1.1mm, black] (0,3) -- (2,3);
\draw [line width=1.1mm, black] (0,4) -- (4,4);

\draw [line width=1.1mm, black] (0,0) -- (0,4);
\draw [line width=1.1mm, black] (1,0) -- (1,4);
\draw [line width=1.1mm, black] (2,1) -- (2,3);
\draw [line width=1.1mm, black] (3,0) -- (3,4);
\draw [line width=1.1mm, black] (4,0) -- (4,4);

\ra{1}{3}

\la{1}{0}

\ua{2}{0}
\ua{3}{0}
\ua{3}{2}

\da{0}{3}

\ix{0}{0}
\ix{0}{1}
\ix{3}{1}
\ix{3}{3}
\end{tikzpicture}}\end{minipage}\quad
\begin{minipage}{.74\textwidth}The indices explain how cells are addressed on the board. Boxes are defined by drawing thicker boundaries.  Up to symmetry, all different 2-boxes and 3-boxes are used. There are four other types of 4-boxes: 
\scalebox{.2}{\begin{tikzpicture}[scale=0.75]
\draw [step =1, black] (0,0) grid (2,2);
\draw [line width=1.1mm, black] (0,0) -- (2,0);
\draw [line width=1.1mm, black] (2,0) -- (2,2);
\draw [line width=1.1mm, black] (0,0) -- (0,2);
\draw [line width=1.1mm, black] (0,2) -- (2,2);
\end{tikzpicture}}, \scalebox{.2}{\begin{tikzpicture}[scale=0.75]
\draw [step =1, black] (0,0) grid (3,2);
\draw [line width=1.1mm, black] (0,0) -- (3,0);
\draw [line width=1.1mm, black] (3,0) -- (3,1);
\draw [line width=1.1mm, black] (0,0) -- (0,1);
\draw [line width=1.1mm, black] (0,1) -- (1,1);
\draw [line width=1.1mm, black] (1,1) -- (1,2);
\draw [line width=1.1mm, black] (1,2) -- (2,2);
\draw [line width=1.1mm, black] (2,1) -- (2,2);
\draw [line width=1.1mm, black] (2,1) -- (3,1);
\end{tikzpicture}}, \scalebox{.2}{\begin{tikzpicture}[scale=0.75]
\draw [step =1, black] (0,0) grid (4,2);
\draw [line width=1.1mm, black] (0,0) -- (4,0);
\draw [line width=1.1mm, black] (4,0) -- (4,1);
\draw [line width=1.1mm, black] (0,0) -- (0,1);
\draw [line width=1.1mm, black] (0,1) -- (4,1);
\end{tikzpicture}}, \scalebox{.2}{\begin{tikzpicture}[scale=0.75]
		\draw [step =1, black] (0,0) grid (3,2);
		\draw [line width=1.1mm, black] (0,2) -- (2,2);
		\draw [line width=1.1mm, black] (2,1) -- (3,1);
		\draw [line width=1.1mm, black] (1,0) -- (3,0);
		\draw [line width=1.1mm, black] (0,1) -- (1,1);
		\draw [line width=1.1mm, black] (0,1) -- (0,2);
		\draw [line width=1.1mm, black] (2,1) -- (2,2);
		\draw [line width=1.1mm, black] (1,1) -- (1,0);
		\draw [line width=1.1mm, black] (3,1) -- (3,0);
\end{tikzpicture}}. A typical reasoning is: Consider cell $(3,1)$. We cannot leave the board, which excludes $\RA$. The preset 2-box excludes $\DA,\UA$. Hence, $\omega(3,1)=\LA$. Similarly,   $\omega(3,3)=\LA$, so  $\omega(2,3)=\DA$, etc.
\end{minipage}
\end{center}
\caption{Example of a $4 \times 4$ Roma game board, showing the main ingredients of a Roma puzzle and its presentation throughout this paper.}
\label{fig:RomaEx}
\end{figure}

\subsection*{A Derived Decision Problem}

An instance of \textsc{Roma} $\mathcal{R}$ consists of an $n \times n$ grid of cells $\mathcal{C}$ with $\mathcal{C} = \{c_{i,j} \mid i,j \in  [n]\}$. The preset entries of the instance are defined by a partial function $\rho \colon \mathcal{C} \to \{\circ, \RA, \LA, \UA, \DA\}$, where only one cell, the Roma-cell $c_{\mathcal{R}}\in\mathcal{C}$, can be assigned with $\circ$. We assume that $|\rho^{-1}(\circ)|=1$. 
This leaves a set $\mathcal{E}_{\mathcal{R}}$ of empty cells, for which $\rho$ is not defined. The boxes of an instance are given by a set $\mathcal{B}_{\mathcal{R}}$ 
and a total function $\beta \colon \mathcal{C} \to \mathcal{B}_{\mathcal{R}}$, where only up to four cells can be sorted into one box. For brevity, we call a box with $c$ cells a $c$-box. 
A cell can only be sorted into a non-empty box if it is a true neighbor of one of the cells already sorted into that box. A \emph{solution} to an instance is an assignment $\omega:\mathcal{C}\to  \{\circ, \RA, \LA, \UA, \DA\}$, which is a total function that coincides with $\rho$ whenever $\rho$ is defined and that is valid in the sense described next.

From an assignment $\omega$, we can derive a directed graph $G(\omega)=(V,E)$ as follows: $V=\mathcal{C}$. Let $c_{i,j},c_{\ell,k}\in V$.
There is a directed edge $(c_{i,j},c_{\ell,k})\in E$ if and only if one of the following four conditions is satisfied:
\begin{itemize}
\item $\ell=i$ and $k=j+1$ and $\omega(c_{i,j})=\UA$; or: 
$\ell=i$ and $k=j-1$ and $\omega(c_{i,j})=\DA$; or: 
\item $\ell=i+1$ and $k=j$ and $\omega(c_{i,j})=\RA$; or: 
$\ell=i-1$ and $k=j$ and $\omega(c_{i,j})=\LA$.
\end{itemize}

An assignment $\omega$ is called \emph{valid} if the following two conditions are met:
\begin{description}
\item[Box condition.] There is no box to which $\omega$ assigns the same arrow twice, or, more formally:
$$\forall c_{i,j},c_{\ell,k}\in \mathcal{C}:(c_{i,j}\neq c_{\ell,k}\land \beta(c_{i,j})=\beta(c_{\ell,k}))\implies \omega(c_{i,j})\neq\omega(c_{\ell,k})\,.$$
\item[Graph condition.]  $G(\omega)$ is acyclic, weakly connected and contains a unique vertex of out-degree zero, namely~$c_{\mathcal{R}}$.
\end{description}

In particular, the graph condition rules out assignments with $\omega(c_{0,0})=\LA$, because then $c_{0,0}$ has out-degree zero, or with 
$\omega(c_{0,0})=\UA$ and with $\omega(c_{0,1})=\DA$, because then the graph would be neither acyclic nor weakly connected. We next show that all paths lead to Rome if the formulated conditions are met.

\begin{lemma} \label{lem:all-paths-to-Rome}
From each vertex, there is a unique directed path to~$c_{\mathcal{R}}$.
\end{lemma}

\begin{proof}
As any assignment can define for any $c_{i,j}\in V$ at most one $c_{\ell,k}$ such that $(c_{i,j},c_{\ell,k})\in E$, each vertex 
has maximum out-degree of one. This already implies that, from each vertex, there exists at most one directed path to~$c_{\mathcal{R}}$.   The graph condition then tells us that indeed all vertices but the Roma-cell~$c_{\mathcal{R}}$ have out-degree exactly one. 
Let $v_0\in V$, $v_0\neq  c_{\mathcal{R}}$, be arbitrary. As $G(\omega)$ is weakly connected, there exists a sequence of vertices $v_0,v_1,v_2,\dots,v_k$ with $v_k=c_{\mathcal{R}}$ and, for each $i=1,\dots,k$, either $(v_{i-1},v_i)\in E$ or $(v_i,v_{i-1})\in E$ (but not both because of acyclicity). As $v_k$ has out-degree zero, $(v_{k-1},v_k)\in E$. As $v_{k-1}$ has out-degree one, $(v_{k-2},v_{k-1})$ is enforced.
This argument propagates inductively, so that finally $(v_{i-1},v_i)\in E$ for all  $i=1,\dots,k$ can be concluded, i.e., there exists a  directed path from $v_0$ to~$c_{\mathcal{R}}$. 
\end{proof}

\noindent
Given an instance $\mathcal{R}=(n,\rho,\beta)$ of \textsc{Roma}, the question is if there exists a solution or not.
If we want to explicitly mention the board dimensions, we speak of an $n\times n$-\textsc{Roma} puzzle.

Henceforth, we will refer to the information about an element of $\{\UA, \DA, \LA, \RA\}$ as a signal or flow, whereas ``signal'' describes an information which can be passed on to another cell by utilizing the rules of Roma (for example the special relationship between cells contained within the same box) and ``flow'' describes the path which is followed when moving one cell at a time in the direction the arrow contained in each given cell points in. Each flow needs to end in the Roma-cell in order for the assignment of an instance to be valid.

\section{Computational Hardness Results}

In this section, we present our first main result, which is the following one.
\begin{theorem}\label{thm:NPRoma}
\textsc{Roma} is \NP-complete.
\end{theorem}

As each cell is filled by one out of four possible directions, a solution can be described with at most $2n^2$ bits for a concrete Roma puzzle with $n \times n$ cells. Hence,  \textsc{Roma} is in \NP.
The main part of the section is devoted to a proof sketch for its \NP-hardness.
At the end, we derive several other conclusions from the specific nature of the reduction.

Recall that Roma is played on an $n\times n$ board.
The \NP-hardness of the \textsc{Roma} puzzle is proven by a reduction from \textsc{Planar 3SAT}.
According to Lichtenstein~\cite{DBLP:journals/siamcomp/Lichtenstein82}, a 3-SAT formula $\varphi$ is \emph{planar} if the
graph $G(\varphi)=(V,E)$ admits an embedding in the plane, where $V=X\cup C$, with $X$ being the variables of $\varphi$ and $C$ being its clauses, and $E$ contains two types of edges: (a) incidence edges: $xc\in E$ if variable~$x$ occurs (either positive or negated) in clause~$c$, (b) cycle edges: $G[X]$ is a cycle. Instances of \textsc{Planar 3SAT} are planar 3-SAT formulas. It is worth mentioning that in Lichtenstein's construction, the graph $G(\varphi)$ is of bounded degree as can be observed in~\autoref{fig:L-crossing}, where the crossover-gadget of the construction is depicted.

We will show how to construct a \textsc{Roma} puzzle $R(\varphi)$ from a given planar 3-SAT formula~$\varphi$ in polynomial time.
This implicitly assumes a planar embedding of the graph $G(\varphi)=(V,E)$.
The edges of $G(\varphi)$ require the distribution of a  signal between the variable and clause gadgets that we describe below.
As we have to model these edges in a discretized fashion in $R(\varphi)$, we need a simple gadget that allows to turn signals by 90 degrees.
By describing this gadget, it should also be made clear what \emph{signal} means in our construction. Note that gadgets described in the following will be embedded in a bigger Roma board, hence arrows can leave the gadget. We will later take care of those arrows by leading their flow to the Roma-cell.

\vspace{5pt}
\noindent
First, we create an L-shaped 4-box. In the context of this explanation, we set the lower left cell of this box as $c_{0,0}$ with $\rho(c_{0,0}) = \UA$ and $\rho(c_{0,1}) =\LA$. This\linebreak[4]
\begin{minipage}{.8\textwidth} means that $\beta(c_{0,0}) = \beta(c_{0,1}) = \beta(c_{0,2}) = \beta(c_{1,0}) = b_{0}$. Then, we create three additional boxes with $\beta(c_{2,0}) = b_{1}$; $\beta(c_{1,2}) = b_{2}$; $\beta(c_{0,3}) = \beta(c_{1,3}) = b_{3}$. Note that we may extend $b_{1}$ and $b_{3}$ to contain additional cells if needed. Lastly, we set $\rho(c_{1,3}) =\LA$ and $\rho(c_{1,2}) = \UA$. This gives the construction called \emph{corner-gadget}, depicted on the right. The cells showing their indices are still to be set during a play. Note, that this gadget will be utilized as part of the variable-gadget.
\end{minipage}
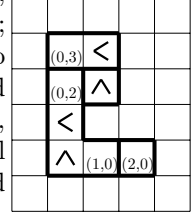
\begin{minipage}{.2\textwidth}\scalebox{.6}{\begin{tikzpicture}[scale=0.79]
\draw [step =1, black] (-1,-1) grid (4,5);

\draw [line width=0.9mm, black] (0,0) -- (3,0);
\draw [line width=0.9mm, black] (1,1) -- (3,1);
\draw [line width=0.9mm, black] (1,2) -- (2,2);
\draw [line width=0.9mm, black] (0,3) -- (2,3);
\draw [line width=0.9mm, black] (0,4) -- (2,4);
\draw [line width=0.9mm, black] (0,0) -- (0,4);
\draw [line width=0.9mm, black] (1,1) -- (1,3);
\draw [line width=0.9mm, black] (2,0) -- (2,1);
\draw [line width=0.9mm, black] (2,2) -- (2,4);
\draw [line width=0.9mm, black] (3,0) -- (3,1);

\la{0}{1}
\la{1}{3}

\ua{0}{0}
\ua{1}{2}

\ix{1}{0}
\ix{2}{0}
\ix{0}{2}
\ix{0}{3}
\end{tikzpicture}}
\end{minipage}

\vspace{5pt}
Now and in the following, assume that $\omega$ is an assignment that resolves the \textsc{Roma} puzzle, where this gadget is a piece of. We can show the following claim:
\begin{lemma}
	\label{lem:corner-gadget}
$(\omega(c_{2,0}) =\LA)$ implies $(\omega(c_{0,3}) = \UA)$ and $(\omega(c_{0,3}) = \DA)$ implies $(\omega(c_{1,0}) =\RA)$.
\end{lemma}

Hence, a path entering the gadget from the right lower end will cause an upward direction at the upper left end, while a downward path at the upper left end will cause a right direction to be taken at the right lower end. We will show-case the reasoning with this lemma, but rather only state claims below.

\begin{proof}
We need to prevent closed cycles, since otherwise there would be cells, the flow of which would not reach the Roma-cell. This 
leads to:\\
$(\omega(c_{2,0}) =\LA) \rightarrow (\omega(c_{1,0}) = \DA) \rightarrow (\omega(c_{0,2}) =\RA) \rightarrow (\omega(c_{0,3}) = \UA)$
and\\
$(\omega(c_{0,3}) = \DA) \rightarrow (\omega(c_{0,2}) = \DA) \rightarrow (\omega(c_{1,0}) =\RA)$. Further, $\omega(c_{2,0}) \ne\LA$. 
\end{proof}

We also need to move a signal along in one direction. This can be done with a \emph{straight-line gadget}, described next.\\[3pt]We create a 4-box, all cells of which are located in the same row. In the  context of this explanation, we set the leftmost cell of this box as $c_{0,0}$. This means that\linebreak[4]
\begin{minipage}{.75\textwidth}
 $\beta(c_{0,0}) = \beta(c_{1,0}) = \beta(c_{2,0}) = \beta(c_{3,0}) = b_{0}$. We set $\omega(c_{1,0}) =\LA$ and $\omega(c_{2,0}) =\RA$. A box such as this one will henceforth be referred to as a \emph{conductor-box}. In order to complete our gadget, we place two additional conductor-boxes on top of the first one. The cells showing their indices are still to be set during a play. 
\end{minipage}
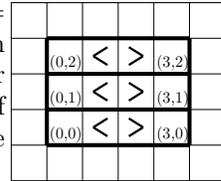
\begin{minipage}{.25\textwidth}\scalebox{.6}{\begin{tikzpicture}[scale=0.79]

\draw [step =1, black] (-1,-1) grid (5,4);

\draw [line width=0.9mm, black] (0,0) -- (4,0);
\draw [line width=0.9mm, black] (0,1) -- (4,1);
\draw [line width=0.9mm, black] (0,2) -- (4,2);
\draw [line width=0.9mm, black] (0,3) -- (4,3);
\draw [line width=0.9mm, black] (0,0) -- (0,3);
\draw [line width=0.9mm, black] (4,0) -- (4,3);

\ra{2}{0}
\ra{2}{1}
\ra{2}{2}

\la{1}{0}
\la{1}{1}
\la{1}{2}

\ix{0}{0}
\ix{0}{1}
\ix{0}{2}
\ix{3}{0}
\ix{3}{1}
\ix{3}{2}
\end{tikzpicture}}
\end{minipage}

\begin{lemma}\label{lem:straight-line gadget}
If any cell in a straight-line gadget is validly assigned, there is only one valid way to assign the other empty cells.
\end{lemma}

\begin{proof}
Assigning any cell within this gadget leads, by utilizing the rules of Roma, to the following implications. W.l.o.g., we start by assigning $c_{0,0}$: 
$(\omega(c_{0,0}) = \DA) \rightarrow (\omega(c_{3,0}) = \UA) \rightarrow (\omega(c_{3,1}) = \UA) \rightarrow (\omega(c_{3,2}) = \UA) \rightarrow (\omega(c_{0,2}) = \DA) \rightarrow (\omega(c_{0,1}) = \DA) \rightarrow (\omega(c_{0,0}) = \DA)$ and
$(\omega(c_{0,0}) = \UA) \rightarrow (\omega(c_{0,1}) = \UA) \rightarrow (\omega(c_{0,2}) = \UA) \rightarrow (\omega(c_{3,2}) = \DA) \rightarrow (\omega(c_{3,1}) = \DA) \rightarrow (\omega(c_{3,0}) = \DA) \rightarrow (\omega(c_{0,0}) = \UA)$.
\end{proof}

By adding more conductor-boxes, straight lines of arbitrary length can be built.
We are now coming to the vertices of $G(\varphi)$. As they also have higher degree (although the degree can be always assumed to be bounded),
we need some gadgets to fan signals out. We  need to be able to fan a signal out in order to deliver \linebreak[4]
\begin{minipage}{.52\textwidth}
  an information regarding a variable to multiple gadgets representing clauses containing the said variable. This is done utilizing the  \emph{fanout-gadget} displayed on the right. It is based on the idea of straight-line gadgets, so that we  conclude:
\end{minipage}
\begin{minipage}{.48\textwidth}\scalebox{.6}{\begin{tikzpicture}[scale=0.79]

\draw [step =1, black] (-1,-1) grid (11,4);

\draw [line width=0.9mm, black] (0,0) -- (4,0);
\draw [line width=0.9mm, black] (6,0) -- (10,0);
\draw [line width=0.9mm, black] (0,1) -- (10,1);
\draw [line width=0.9mm, black] (0,2) -- (10,2);
\draw [line width=0.9mm, black] (0,3) -- (4,3);
\draw [line width=0.9mm, black] (6,3) -- (10,3);
\draw [line width=0.9mm, black] (0,0) -- (0,3);
\draw [line width=0.9mm, black] (3,1) -- (3,2);
\draw [line width=0.9mm, black] (4,0) -- (4,1);
\draw [line width=0.9mm, black] (4,2) -- (4,3);
\draw [line width=0.9mm, black] (6,0) -- (6,1);
\draw [line width=0.9mm, black] (6,2) -- (6,3);
\draw [line width=0.9mm, black] (7,1) -- (7,2);
\draw [line width=0.9mm, black] (10,0) -- (10,3);

\ra{2}{0}
\ra{2}{1}
\ra{2}{2}
\ra{5}{1}
\ra{8}{0}
\ra{8}{1}
\ra{8}{2}

\la{1}{0}
\la{1}{1}
\la{1}{2}
\la{4}{1}
\la{7}{0}
\la{7}{1}
\la{7}{2}

\ix{0}{0}
\ix{0}{1}
\ix{0}{2}
\ix{3}{0}
\ix{3}{1}
\ix{3}{2}
\ix{6}{0}
\ix{6}{1}
\ix{6}{2}
\ix{9}{0}
\ix{9}{1}
\ix{9}{2}
\end{tikzpicture}}
\end{minipage}

\begin{lemma}\label{lem:fanout}
If any empty cell in a fanout-gadget is assigned, there is only one valid way to assign the other empty cells.
\end{lemma}

Note that an arbitrary number of variations of fanout-gagdets can be connected and utilized in order to transmit a signal horizontally:

\begin{minipage}{.48\textwidth}\scalebox{.55}{\begin{tikzpicture}[scale=0.79]

\draw [step =1, black] (-1,-0) grid (23,3);

\draw [line width=0.9mm, black] (0,0) -- (4,0);
\draw [line width=0.9mm, black] (6,0) -- (10,0);
\draw [line width=0.9mm, black] (12,0) -- (16,0);
\draw [line width=0.9mm, black] (18,0) -- (22,0);
\draw [line width=0.9mm, black] (0,1) -- (22,1);
\draw [line width=0.9mm, black] (0,2) -- (22,2);
\draw [line width=0.9mm, black] (0,3) -- (4,3);
\draw [line width=0.9mm, black] (6,3) -- (10,3);
\draw [line width=0.9mm, black] (12,3) -- (16,3);
\draw [line width=0.9mm, black] (18,3) -- (22,3);
\draw [line width=0.9mm, black] (0,0) -- (0,3);
\draw [line width=0.9mm, black] (3,1) -- (3,2);
\draw [line width=0.9mm, black] (4,0) -- (4,1);
\draw [line width=0.9mm, black] (4,2) -- (4,3);
\draw [line width=0.9mm, black] (6,0) -- (6,1);
\draw [line width=0.9mm, black] (6,2) -- (6,3);
\draw [line width=0.9mm, black] (7,1) -- (7,2);
\draw [line width=0.9mm, black] (9,1) -- (9,2);
\draw [line width=0.9mm, black] (10,0) -- (10,1);
\draw [line width=0.9mm, black] (10,2) -- (10,3);
\draw [line width=0.9mm, black] (12,0) -- (12,1);
\draw [line width=0.9mm, black] (12,2) -- (12,3);
\draw [line width=0.9mm, black] (13,1) -- (13,2);
\draw [line width=0.9mm, black] (15,1) -- (15,2);
\draw [line width=0.9mm, black] (16,0) -- (16,1);
\draw [line width=0.9mm, black] (16,2) -- (16,3);
\draw [line width=0.9mm, black] (18,0) -- (18,1);
\draw [line width=0.9mm, black] (18,2) -- (18,3);
\draw [line width=0.9mm, black] (19,1) -- (19,2);
\draw [line width=0.9mm, black] (21,1) -- (21,2);
\draw [line width=0.9mm, black] (22,0) -- (22,3);

\ra{2}{0}
\ra{2}{1}
\ra{2}{2}
\ra{5}{1}
\ra{8}{0}
\ra{8}{1}
\ra{8}{2}
\ra{11}{1}
\ra{14}{0}
\ra{14}{1}
\ra{14}{2}
\ra{17}{1}
\ra{20}{0}
\ra{20}{1}
\ra{20}{2}

\la{1}{0}
\la{1}{1}
\la{1}{2}
\la{4}{1}
\la{7}{0}
\la{7}{1}
\la{7}{2}
\la{10}{1}
\la{13}{0}
\la{13}{1}
\la{13}{2}
\la{16}{1}
\la{19}{0}
\la{19}{1}
\la{19}{2}

\end{tikzpicture}}
\end{minipage}
\\

So far, we mainly constructed geometric gadgets, but these are the building blocks of the proper formula gadgets that we describe next.
Let us mention one more geometric detail: In $G(\varphi)$, all vertices have been connected via a cycle. In our construction, we actually only need a connection via a kind of path which we refert to as the \emph{core-line} in the following.
Next, we describe the logical gadgets, which are gadgets for setting variables, literals and clauses.

The variable-gadget is described in \autoref{fig:vgscb}. The picture contains quite a number of preset 1-box cells in the middle, most of which are not necessary for the gadgetry itself. Its sole purpose is to form a proper Roma puzzle and to guarantee that there is only one possible way to solve the Roma puzzle in case the given Boolean formula was uniquely satisfiable. 
The essential preset cells for the variable-gadget are shown in \autoref{fig:vgf}. This gives some empty space, but there is clearly far more empty space to be filled between the gadgets, when we assemble the whole construction of a given Boolean formula. 
How to fill this space is explained in more details at the end of this section.


\begin{figure}[tb]
We start with a fanout-gadget at the bottom. Next, the lower left part of this gadget will be connected to a modified corner-gadget, which does not\linebreak[4]
\begin{minipage}{.47\textwidth}  have the lower right 1-box as shown above, and the top 2-box will be part of the lower left 4-box of the fanout-gadget instead. We now connect a corner-gadget, modified in a slightly different manner, to the lower right part. Lastly, we copy this construction, rotate it by $180^\circ$ and connect the two by connecting the loose ends of the corner-gadgets, which allows for additional 1-boxes to create a flow from the right end of the gadget to its left end as part of the core-line, highlighted in light red.
\end{minipage}\ \begin{minipage}{.52\textwidth}
\scalebox{.45}{
\begin{tikzpicture}[scale=0.94]
\draw [step =1, black] (-1,-1) grid (14,12);
\draw [step =1, line width=1.5mm, black] (4,3) grid (9,6);
\draw [step =1, line width=1.5mm, black] (9,3) grid (11,5);
\draw [step =1, line width=1.5mm, black] (2,7) grid (9,8);

\draw [line width=1.5mm, black] (8,2) -- (8,3);
\draw [line width=1.5mm, black] (5,8) -- (5,9);

\filldraw[fill=red!20 ](0,6) rectangle (13,7);
\draw [line width=1.5mm, black] (3,0) -- (7,0);
\draw [line width=1.5mm, black] (9,0) -- (13,0);
\draw [line width=1.5mm, black] (3,1) -- (13,1);
\draw [line width=1.5mm, black] (2,2) -- (13,2);
\draw [line width=1.5mm, black] (3,3) -- (7,3);
\draw [line width=1.5mm, black] (9,3) -- (13,3);
\draw [line width=1.5mm, black] (2,4) -- (3,4);
\draw [line width=1.5mm, black] (11,4) -- (12,4);
\draw [line width=1.5mm, black] (0,5) -- (3,5);
\draw [line width=1.5mm, black] (9,5) -- (12,5);
\draw [line width=1.5mm, black] (1,6) -- (9,6);
\draw [line width=1.5mm, black] (10,6) -- (13,6);
\draw [line width=1.5mm, black] (1,7) -- (9,7);
\draw [line width=1.5mm, black] (10,7) -- (13,7);
\draw [line width=1.5mm, black] (0,8) -- (4,8);
\draw [line width=1.5mm, black] (6,8) -- (10,8);
\draw [line width=1.5mm, black] (0,9) -- (11,9);
\draw [line width=1.5mm, black] (0,10) -- (10,10);
\draw [line width=1.5mm, black] (0,11) -- (4,11);
\draw [line width=1.5mm, black] (6,11) -- (10,11);
\draw [line width=1.5mm, black] (0,5) -- (0,11);
\draw [line width=1.5mm, black] (1,6) -- (1,8);
\draw [line width=1.5mm, black] (2,2) -- (2,4);
\draw [line width=1.5mm, black] (2,5) -- (2,8);
\draw [line width=1.5mm, black] (3,0) -- (3,5);
\draw [line width=1.5mm, black] (3,6) -- (3,7);
\draw [line width=1.5mm, black] (3,9) -- (3,10);
\draw [line width=1.5mm, black] (4,3) -- (4,7);
\draw [line width=1.5mm, black] (4,8) -- (4,9);
\draw [line width=1.5mm, black] (4,10) -- (4,11);
\draw [line width=1.5mm, black] (5,6) -- (5,7);
\draw [line width=1.5mm, black] (6,1) -- (6,2);
\draw [line width=1.5mm, black] (6,6) -- (6,7);
\draw [line width=1.5mm, black] (6,8) -- (6,9);
\draw [line width=1.5mm, black] (6,10) -- (6,11);
\draw [line width=1.5mm, black] (7,0) -- (7,1);
\draw [line width=1.5mm, black] (7,2) -- (7,3);
\draw [line width=1.5mm, black] (7,6) -- (7,7);
\draw [line width=1.5mm, black] (7,9) -- (7,10);
\draw [line width=1.5mm, black] (8,6) -- (8,7);
\draw [line width=1.5mm, black] (9,0) -- (9,1);
\draw [line width=1.5mm, black] (9,2) -- (9,3);
\draw [line width=1.5mm, black] (9,5) -- (9,8);
\draw [line width=1.5mm, black] (10,1) -- (10,2);
\draw [line width=1.5mm, black] (10,6) -- (10,11);
\draw [line width=1.5mm, black] (11,3) -- (11,4);
\draw [line width=1.5mm, black] (11,5) -- (11,9);
\draw [line width=1.5mm, black] (12,3) -- (12,5);
\draw [line width=1.5mm, black] (12,6) -- (12,7);
\draw [line width=1.5mm, black] (13,0) -- (13,7);

\ra{5}{0}
\ra{11}{0}
\ra{5}{1}
\ra{8}{1}
\ra{11}{1}
\ra{2}{2}
\ra{5}{2}
\ra{11}{2}
\ra{3}{4}
\ra{12}{4}
\ra{2}{8}
\ra{8}{8}
\ra{2}{9}
\ra{5}{9}
\ra{8}{9}
\ra{2}{10}
\ra{8}{10}

\la{4}{0}
\la{10}{0}
\la{4}{1}
\la{7}{1}
\la{10}{1}
\la{4}{2}
\la{10}{2}
\la{0}{6}
\la{1}{6}
\la{2}{6}
\la{3}{6}
\la{4}{6}
\la{5}{6}
\la{6}{6}
\la{7}{6}
\la{8}{6}
\la{9}{6}
\la{10}{6}
\la{11}{6}
\la{12}{6}
\la{1}{8}
\la{7}{8}
\la{10}{8}
\la{1}{9}
\la{4}{9}
\la{7}{9}
\la{1}{10}
\la{7}{10}

\la{9}{3}
\la{9}{4}
\la{10}{3}
\la{10}{4}
\la{11}{4}

\ua{0}{5}
\ua{9}{5}
\ua{1}{7}
\ua{10}{7}

\ua{4}{5}
\ua{5}{5}
\ua{6}{5}
\ua{7}{5}
\ua{8}{5}
\ua{4}{4}
\ua{5}{4}
\ua{6}{4}
\ua{7}{4}
\ua{8}{4}
\ua{4}{3}
\ua{5}{3}
\ua{6}{3}
\ua{7}{3}
\ua{8}{3}
\ua{7}{2}
\ua{8}{2}

\da{2}{3}
\da{11}{3}
\da{3}{5}
\da{12}{5}

\da{2}{7}
\da{3}{7}
\da{4}{7}
\da{5}{7}
\da{6}{7}
\da{7}{7}
\da{8}{7}
\da{4}{8}
\da{5}{8}

\ix{0}{0}
\ix{3}{0}
\ix{6}{0}
\ix{9}{0}
\ix{12}{0}
\ix{3}{1}
\ix{6}{1}
\ix{9}{1}
\ix{12}{1}
\ix{3}{2}
\ix{6}{2}
\ix{9}{2}
\ix{12}{2}
\ix{3}{3}
\ix{12}{3}
\ix{1}{5}
\ix{2}{5}
\ix{10}{5}
\ix{11}{5}
\ix{0}{7}
\ix{9}{7}
\ix{0}{8}
\ix{3}{8}
\ix{6}{8}
\ix{9}{8}
\ix{0}{9}
\ix{3}{9}
\ix{6}{9}
\ix{9}{9}
\ix{0}{10}
\ix{3}{10}
\ix{6}{10}
\ix{9}{10}
\end{tikzpicture}}
\end{minipage}
\caption{Variable-gadget with 1-boxes in the middle to form the core-line.}
\label{fig:vgscb}
\end{figure}
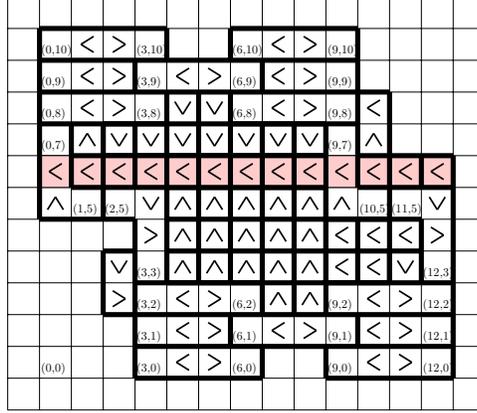

As the variable-gadget is built up from fanout- and corner-gadgets, the next lemma follows directly from~\autoref{lem:corner-gadget} and~\autoref{lem:fanout}.
\begin{lemma}
If any empty cell in a variable-gadget is assigned, there is only one valid way to assign the other empty cells.
\end{lemma}

As each cell within the variable-gadget can only be validly assigned with one of two options, the resulting two possibilities to assign the variable-gadget are depicted in \autoref{fig:vgf}. Clearly, the two options should correspond to setting the variable true or false.

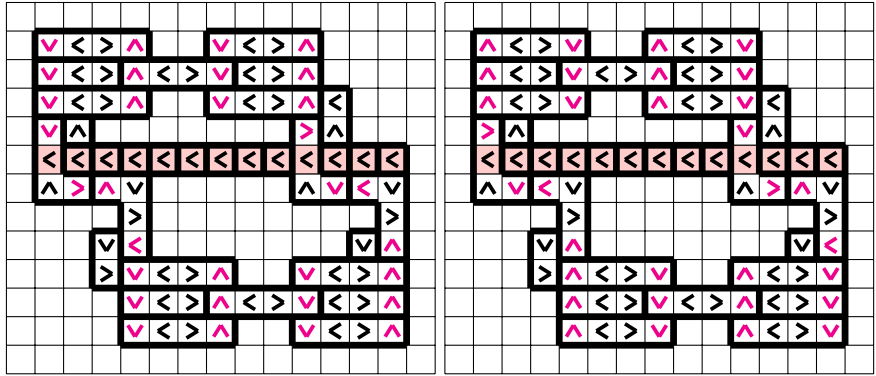
\begin{figure}[tb]
\begin{center}
\begin{tikzpicture}[scale=0.38]
\draw [step =1, black] (-1,-1) grid (14,12);

\filldraw[fill=red!20 ](0,6) rectangle (13,7);

\draw [line width=0.9mm, black] (3,0) -- (7,0);
\draw [line width=0.9mm, black] (9,0) -- (13,0);
\draw [line width=0.9mm, black] (3,1) -- (13,1);
\draw [line width=0.9mm, black] (2,2) -- (13,2);
\draw [line width=0.9mm, black] (3,3) -- (7,3);
\draw [line width=0.9mm, black] (9,3) -- (13,3);
\draw [line width=0.9mm, black] (2,4) -- (3,4);
\draw [line width=0.9mm, black] (11,4) -- (12,4);
\draw [line width=0.9mm, black] (0,5) -- (3,5);
\draw [line width=0.9mm, black] (9,5) -- (12,5);
\draw [line width=0.9mm, black] (1,6) -- (9,6);
\draw [line width=0.9mm, black] (10,6) -- (13,6);
\draw [line width=0.9mm, black] (1,7) -- (9,7);
\draw [line width=0.9mm, black] (10,7) -- (13,7);
\draw [line width=0.9mm, black] (0,8) -- (4,8);
\draw [line width=0.9mm, black] (6,8) -- (10,8);
\draw [line width=0.9mm, black] (0,9) -- (11,9);
\draw [line width=0.9mm, black] (0,10) -- (10,10);
\draw [line width=0.9mm, black] (0,11) -- (4,11);
\draw [line width=0.9mm, black] (6,11) -- (10,11);
\draw [line width=0.9mm, black] (0,5) -- (0,11);
\draw [line width=0.9mm, black] (1,6) -- (1,8);
\draw [line width=0.9mm, black] (2,2) -- (2,4);
\draw [line width=0.9mm, black] (2,5) -- (2,8);
\draw [line width=0.9mm, black] (3,0) -- (3,5);
\draw [line width=0.9mm, black] (3,6) -- (3,7);
\draw [line width=0.9mm, black] (3,9) -- (3,10);
\draw [line width=0.9mm, black] (4,3) -- (4,7);
\draw [line width=0.9mm, black] (4,8) -- (4,9);
\draw [line width=0.9mm, black] (4,10) -- (4,11);
\draw [line width=0.9mm, black] (5,6) -- (5,7);
\draw [line width=0.9mm, black] (6,1) -- (6,2);
\draw [line width=0.9mm, black] (6,6) -- (6,7);
\draw [line width=0.9mm, black] (6,8) -- (6,9);
\draw [line width=0.9mm, black] (6,10) -- (6,11);
\draw [line width=0.9mm, black] (7,0) -- (7,1);
\draw [line width=0.9mm, black] (7,2) -- (7,3);
\draw [line width=0.9mm, black] (7,6) -- (7,7);
\draw [line width=0.9mm, black] (7,9) -- (7,10);
\draw [line width=0.9mm, black] (8,6) -- (8,7);
\draw [line width=0.9mm, black] (9,0) -- (9,1);
\draw [line width=0.9mm, black] (9,2) -- (9,3);
\draw [line width=0.9mm, black] (9,5) -- (9,8);
\draw [line width=0.9mm, black] (10,1) -- (10,2);
\draw [line width=0.9mm, black] (10,6) -- (10,11);
\draw [line width=0.9mm, black] (11,3) -- (11,4);
\draw [line width=0.9mm, black] (11,5) -- (11,9);
\draw [line width=0.9mm, black] (12,3) -- (12,5);
\draw [line width=0.9mm, black] (12,6) -- (12,7);
\draw [line width=0.9mm, black] (13,0) -- (13,7);

\ra{5}{0}
\ra{11}{0}
\ra{5}{1}
\ra{8}{1}
\ra{11}{1}
\ra{2}{2}
\ra{5}{2}
\ra{11}{2}
\ra{3}{4}
\ra{12}{4}
\ra{2}{8}
\ra{8}{8}
\ra{2}{9}
\ra{5}{9}
\ra{8}{9}
\ra{2}{10}
\ra{8}{10}

\la{4}{0}
\la{10}{0}
\la{4}{1}
\la{7}{1}
\la{10}{1}
\la{4}{2}
\la{10}{2}
\la{0}{6}
\la{1}{6}
\la{2}{6}
\la{3}{6}
\la{4}{6}
\la{5}{6}
\la{6}{6}
\la{7}{6}
\la{8}{6}
\la{9}{6}
\la{10}{6}
\la{11}{6}
\la{12}{6}
\la{1}{8}
\la{7}{8}
\la{10}{8}
\la{1}{9}
\la{4}{9}
\la{7}{9}
\la{1}{10}
\la{7}{10}

\ua{0}{5}
\ua{9}{5}
\ua{1}{7}
\ua{10}{7}

\da{2}{3}
\da{11}{3}
\da{3}{5}
\da{12}{5}

\mda{3}{0}
\mua{6}{0}
\mda{9}{0}
\mua{12}{0}
\mda{3}{1}
\mua{6}{1}
\mda{9}{1}
\mua{12}{1}
\mda{3}{2}
\mua{6}{2}
\mda{9}{2}
\mua{12}{2}
\mla{3}{3}
\mua{12}{3}
\mra{1}{5}
\mua{2}{5}
\mda{10}{5}
\mla{11}{5}
\mda{0}{7}
\mra{9}{7}
\mda{0}{8}
\mua{3}{8}
\mda{6}{8}
\mua{9}{8}
\mda{0}{9}
\mua{3}{9}
\mda{6}{9}
\mua{9}{9}
\mda{0}{10}
\mua{3}{10}
\mda{6}{10}
\mua{9}{10}
\end{tikzpicture}\ 
\begin{tikzpicture}[scale=0.38]
\draw [step =1, black] (-1,-1) grid (14,12);

\filldraw[fill=red!20 ](0,6) rectangle (13,7);

\draw [line width=0.9mm, black] (3,0) -- (7,0);
\draw [line width=0.9mm, black] (9,0) -- (13,0);
\draw [line width=0.9mm, black] (3,1) -- (13,1);
\draw [line width=0.9mm, black] (2,2) -- (13,2);
\draw [line width=0.9mm, black] (3,3) -- (7,3);
\draw [line width=0.9mm, black] (9,3) -- (13,3);
\draw [line width=0.9mm, black] (2,4) -- (3,4);
\draw [line width=0.9mm, black] (11,4) -- (12,4);
\draw [line width=0.9mm, black] (0,5) -- (3,5);
\draw [line width=0.9mm, black] (9,5) -- (12,5);
\draw [line width=0.9mm, black] (1,6) -- (9,6);
\draw [line width=0.9mm, black] (10,6) -- (13,6);
\draw [line width=0.9mm, black] (1,7) -- (9,7);
\draw [line width=0.9mm, black] (10,7) -- (13,7);
\draw [line width=0.9mm, black] (0,8) -- (4,8);
\draw [line width=0.9mm, black] (6,8) -- (10,8);
\draw [line width=0.9mm, black] (0,9) -- (11,9);
\draw [line width=0.9mm, black] (0,10) -- (10,10);
\draw [line width=0.9mm, black] (0,11) -- (4,11);
\draw [line width=0.9mm, black] (6,11) -- (10,11);
\draw [line width=0.9mm, black] (0,5) -- (0,11);
\draw [line width=0.9mm, black] (1,6) -- (1,8);
\draw [line width=0.9mm, black] (2,2) -- (2,4);
\draw [line width=0.9mm, black] (2,5) -- (2,8);
\draw [line width=0.9mm, black] (3,0) -- (3,5);
\draw [line width=0.9mm, black] (3,6) -- (3,7);
\draw [line width=0.9mm, black] (3,9) -- (3,10);
\draw [line width=0.9mm, black] (4,3) -- (4,7);
\draw [line width=0.9mm, black] (4,8) -- (4,9);
\draw [line width=0.9mm, black] (4,10) -- (4,11);
\draw [line width=0.9mm, black] (5,6) -- (5,7);
\draw [line width=0.9mm, black] (6,1) -- (6,2);
\draw [line width=0.9mm, black] (6,6) -- (6,7);
\draw [line width=0.9mm, black] (6,8) -- (6,9);
\draw [line width=0.9mm, black] (6,10) -- (6,11);
\draw [line width=0.9mm, black] (7,0) -- (7,1);
\draw [line width=0.9mm, black] (7,2) -- (7,3);
\draw [line width=0.9mm, black] (7,6) -- (7,7);
\draw [line width=0.9mm, black] (7,9) -- (7,10);
\draw [line width=0.9mm, black] (8,6) -- (8,7);
\draw [line width=0.9mm, black] (9,0) -- (9,1);
\draw [line width=0.9mm, black] (9,2) -- (9,3);
\draw [line width=0.9mm, black] (9,5) -- (9,8);
\draw [line width=0.9mm, black] (10,1) -- (10,2);
\draw [line width=0.9mm, black] (10,6) -- (10,11);
\draw [line width=0.9mm, black] (11,3) -- (11,4);
\draw [line width=0.9mm, black] (11,5) -- (11,9);
\draw [line width=0.9mm, black] (12,3) -- (12,5);
\draw [line width=0.9mm, black] (12,6) -- (12,7);
\draw [line width=0.9mm, black] (13,0) -- (13,7);

\ra{5}{0}
\ra{11}{0}
\ra{5}{1}
\ra{8}{1}
\ra{11}{1}
\ra{2}{2}
\ra{5}{2}
\ra{11}{2}
\ra{3}{4}
\ra{12}{4}
\ra{2}{8}
\ra{8}{8}
\ra{2}{9}
\ra{5}{9}
\ra{8}{9}
\ra{2}{10}
\ra{8}{10}

\la{4}{0}
\la{10}{0}
\la{4}{1}
\la{7}{1}
\la{10}{1}
\la{4}{2}
\la{10}{2}
\la{0}{6}
\la{1}{6}
\la{2}{6}
\la{3}{6}
\la{4}{6}
\la{5}{6}
\la{6}{6}
\la{7}{6}
\la{8}{6}
\la{9}{6}
\la{10}{6}
\la{11}{6}
\la{12}{6}
\la{1}{8}
\la{7}{8}
\la{10}{8}
\la{1}{9}
\la{4}{9}
\la{7}{9}
\la{1}{10}
\la{7}{10}

\ua{0}{5}
\ua{9}{5}
\ua{1}{7}
\ua{10}{7}

\da{2}{3}
\da{11}{3}
\da{3}{5}
\da{12}{5}

\mua{3}{0}
\mda{6}{0}
\mua{9}{0}
\mda{12}{0}
\mua{3}{1}
\mda{6}{1}
\mua{9}{1}
\mda{12}{1}
\mua{3}{2}
\mda{6}{2}
\mua{9}{2}
\mda{12}{2}
\mua{3}{3}
\mla{12}{3}
\mda{1}{5}
\mla{2}{5}
\mra{10}{5}
\mua{11}{5}
\mra{0}{7}
\mda{9}{7}
\mua{0}{8}
\mda{3}{8}
\mua{6}{8}
\mda{9}{8}
\mua{0}{9}
\mda{3}{9}
\mua{6}{9}
\mda{9}{9}
\mua{0}{10}
\mda{3}{10}
\mua{6}{10}
\mda{9}{10}
\end{tikzpicture}
\end{center}
\caption{There are exactly two ways to fill in the empty cells of the variable-gadget. The filled-in arrows on the left side are interpreted as setting the variable to true, while the way
the arrows are filled in on the right side should mean that the variable is set to false. Note that this holds for both ways to fill in this gadget and regardless of whether the connected clauses appear above or below the core-line due to rotational symmetry of the fanout-gadgets used in both the variable-gadgets as well as the literal-gadgets (see Figure 4).}
\label{fig:vgf}
\end{figure}

We need a gadget to represent the literals contained in clauses. Literals can either be positive or negative. Positive literals are represented via the  gadget shown in \autoref{fig:LiteralGadgets} on the left-hand side. 
We start the construction, once again, with a fanout-gadget, which will serve as a connection for incoming straight-line gadgets by placing it right below one of its bottom 4-boxes. Only one straight-line gadget will be allowed to any given literal-gadget. We set the bottom left cell of this fanout-gadget to be $c_{0,0}$. Furthermore, we construct two boxes with $\beta(c_{3,3}) = \beta(c_{3,4}) = \beta(c_{3,5}) = \beta(c_{4,5}) = b_{1}$ and $\beta(c_{5,5}) = \beta(c_{6,5}) = \beta(c_{6,4}) = \beta(c_{6,3}) = b_{2}$. We set $\rho(c_{3,4}) = \DA,\; \rho(c_{3,5}) = \UA,\; \rho(c_{6,5}) = \UA$ and $\rho(c_{6,4}) = \DA$. The remaining cells encased within the gadget are filled with three boxes: $\beta(c_{4,2}) = \beta(c_{4,3}) = b_{3},\; \beta(c_{5,2}) = \beta(c_{5,3}) = b_{4}$ and $\beta(c_{4,4}) = \beta(c_{5,4}) = b_{5}$. We set $\rho(c_{4,2}) =\LA,\; \rho(c_{4,3}) = \DA,\; \rho(c_{5,2}) =\RA,\; \rho(c_{5,3}) = \DA,\; \rho(c_{4,4}) =\LA$ and $\rho(c_{5,4}) =\RA$. Lastly, we place pre-filled 1-boxes as shown.
Note that $\omega(c_{6,3}) =\LA$ allows the flow coming from the 1-boxes to leave the gadget downwards, while $\omega(c_{6,3}) =\RA$ funnels it back into the rest of the 1-boxes. This can be seen by following the flow of the 1-boxes in Figure~\ref{fig:LiteralGadgets} and will be relevant when constructing a clause-gadget.

Negative literals are represented via the gadget displayed in \autoref{fig:LiteralGadgets} on the right-hand side. 
We start the construction like the one for positive literals, but instead of placing 10 1-boxes on the right side, we use a somewhat different design on the left side.

\begin{figure}[tb]
\begin{center}
\scalebox{.55}{\begin{tikzpicture}[scale=0.8]
\draw [step =1, black] (-1,-1) grid (11,9);

\draw [line width=1.5mm, black] (0,0) -- (4,0);
\draw [line width=1.5mm, black] (6,0) -- (10,0);
\draw [line width=1.5mm, black] (0,1) -- (10,1);
\draw [line width=1.5mm, black] (0,2) -- (10,2);
\draw [line width=1.5mm, black] (0,3) -- (4,3);
\draw [line width=1.5mm, black] (6,3) -- (10,3);
\draw [line width=1.5mm, black] (4,4) -- (6,4);
\draw [line width=1.5mm, black] (7,4) -- (9,4);
\draw [line width=1.5mm, black] (4,5) -- (6,5);
\draw [line width=1.5mm, black] (7,5) -- (9,5);
\draw [line width=1.5mm, black] (3,6) -- (9,6);
\draw [line width=1.5mm, black] (7,7) -- (9,7);
\draw [line width=1.5mm, black] (7,8) -- (9,8);
\draw [line width=1.5mm, black] (0,0) -- (0,3);
\draw [line width=1.5mm, black] (3,1) -- (3,2);
\draw [line width=1.5mm, black] (3,3) -- (3,6);
\draw [line width=1.5mm, black] (4,0) -- (4,1);
\draw [line width=1.5mm, black] (4,2) -- (4,5);
\draw [line width=1.5mm, black] (5,2) -- (5,4);
\draw [line width=1.5mm, black] (5,5) -- (5,6);
\draw [line width=1.5mm, black] (6,0) -- (6,1);
\draw [line width=1.5mm, black] (6,2) -- (6,5);
\draw [line width=1.5mm, black] (7,1) -- (7,2);
\draw [line width=1.5mm, black] (7,3) -- (7,8);
\draw [line width=1.5mm, black] (8,3) -- (8,8);
\draw [line width=1.5mm, black] (9,3) -- (9,8);
\draw [line width=1.5mm, black] (10,0) -- (10,3);

\ra{2}{0}
\ra{8}{0}
\ra{2}{1}
\ra{5}{1}
\ra{8}{1}
\ra{2}{2}
\ra{5}{2}
\ra{8}{2}
\ra{7}{3}
\ra{5}{4}
\ra{8}{7}

\la{1}{0}
\la{7}{0}
\la{1}{1}
\la{4}{1}
\la{7}{1}
\la{1}{2}
\la{4}{2}
\la{7}{2}
\la{4}{4}
\la{7}{4}

\ua{8}{3}
\ua{8}{4}
\ua{3}{5}
\ua{6}{5}
\ua{8}{5}
\ua{8}{6}

\da{4}{3}
\da{5}{3}
\da{3}{4}
\da{6}{4}
\da{7}{5}
\da{7}{6}
\da{7}{7}

\ix{0}{0}
\ix{3}{0}
\ix{6}{0}
\ix{9}{0}
\ix{0}{1}
\ix{3}{1}
\ix{6}{1}
\ix{9}{1}
\ix{0}{2}
\ix{3}{2}
\ix{6}{2}
\ix{9}{2}
\ix{3}{3}
\ix{6}{3}
\ix{4}{5}
\ix{5}{5}
\end{tikzpicture}}\quad\scalebox{.55}{\begin{tikzpicture}[scale=0.8]
\draw [step =1, black] (-1,-1) grid (11,9);

\draw [line width=1.5mm, black] (0,0) -- (4,0);
\draw [line width=1.5mm, black] (6,0) -- (10,0);
\draw [line width=1.5mm, black] (0,1) -- (10,1);
\draw [line width=1.5mm, black] (0,2) -- (10,2);
\draw [line width=1.5mm, black] (0,3) -- (4,3);
\draw [line width=1.5mm, black] (6,3) -- (10,3);
\draw [line width=1.5mm, black] (2,4) -- (3,4);
\draw [line width=1.5mm, black] (4,4) -- (6,4);
\draw [line width=1.5mm, black] (2,5) -- (3,5);
\draw [line width=1.5mm, black] (4,5) -- (6,5);
\draw [line width=1.5mm, black] (2,6) -- (7,6);
\draw [line width=1.5mm, black] (2,7) -- (5,7);
\draw [line width=1.5mm, black] (6,7) -- (7,7);
\draw [line width=1.5mm, black] (2,8) -- (7,8);
\draw [line width=1.5mm, black] (0,0) -- (0,3);
\draw [line width=1.5mm, black] (2,3) -- (2,7);
\draw [line width=1.5mm, black] (2,7) -- (2,8);
\draw [line width=1.5mm, black] (3,1) -- (3,2);
\draw [line width=1.5mm, black] (3,3) -- (3,7);
\draw [line width=1.5mm, black] (3,7) -- (3,8);
\draw [line width=1.5mm, black] (4,0) -- (4,1);
\draw [line width=1.5mm, black] (4,2) -- (4,5);
\draw [line width=1.5mm, black] (4,6) -- (4,8);
\draw [line width=1.5mm, black] (5,2) -- (5,4);
\draw [line width=1.5mm, black] (5,5) -- (5,8);
\draw [line width=1.5mm, black] (6,0) -- (6,1);
\draw [line width=1.5mm, black] (6,2) -- (6,5);
\draw [line width=1.5mm, black] (6,6) -- (6,8);
\draw [line width=1.5mm, black] (7,1) -- (7,2);
\draw [line width=1.5mm, black] (7,3) -- (7,8);
\draw [line width=1.5mm, black] (10,0) -- (10,3);

\ra{2}{0}
\ra{8}{0}
\ra{2}{1}
\ra{5}{1}
\ra{8}{1}
\ra{2}{2}
\ra{5}{2}
\ra{8}{2}
\ra{2}{4}
\ra{5}{4}
\ra{2}{7}
\ra{3}{7}
\ra{5}{7}
\ra{6}{7}

\la{1}{0}
\la{7}{0}
\la{1}{1}
\la{4}{1}
\la{7}{1}
\la{1}{2}
\la{4}{2}
\la{7}{2}
\la{2}{3}
\la{4}{4}
\la{3}{6}

\ua{3}{5}
\ua{6}{5}
\ua{5}{6}
\ua{6}{6}

\da{4}{3}
\da{5}{3}
\da{3}{4}
\da{6}{4}
\da{2}{5}
\da{2}{6}
\da{4}{6}
\da{4}{7}

\ix{0}{0}
\ix{3}{0}
\ix{6}{0}
\ix{9}{0}
\ix{0}{1}
\ix{3}{1}
\ix{6}{1}
\ix{9}{1}
\ix{0}{2}
\ix{3}{2}
\ix{6}{2}
\ix{9}{2}
\ix{3}{3}
\ix{6}{3}
\ix{4}{5}
\ix{5}{5}
\end{tikzpicture}}
\end{center}
\caption{Variables may occur either positive (on the left) or negated (on the right). This leads to slightly different shapes of the corresponding literal-gadgets that form the basis of the clause-gadgets.}
\label{fig:LiteralGadgets}
\end{figure}
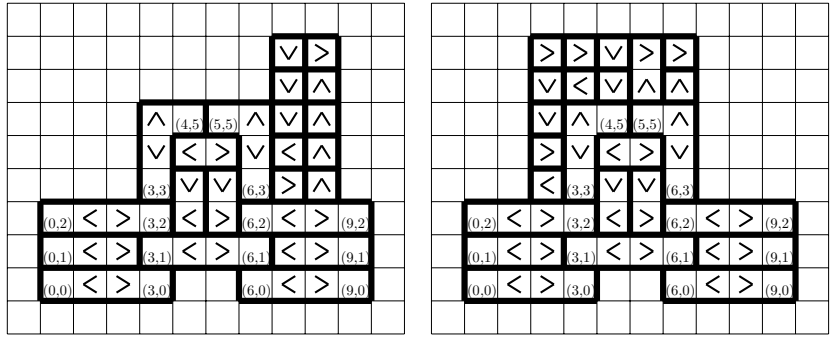

\begin{lemma}\label{lem:literal-gadgets}
If any empty cell in a literal-gadget is validly assigned, there is only one valid way to assign the other empty cells.
\end{lemma}

\begin{proof}
We consider positive literal-gadgets only; the proof for negative literal-gadgets is analogous. 
The statement for the cells belonging to the fanout-gadget follows from \autoref{lem:fanout}. The statement for $c_{3,3}, c_{4,5}, c_{5,5}$ and $c_{6,3}$ can be shown by utilizing the rules of Roma. More precisely, we find:
$(\omega(c_{6,2}) = \DA) \rightarrow (\omega(c_{3,2}) = \UA) \rightarrow (\omega(c_{3,3}) =\LA) \rightarrow (\omega(c_{4,5}) =\RA) \rightarrow (\omega(c_{5,5}) =\RA) \rightarrow (\omega(c_{6,3}) =\LA) \rightarrow (\omega(c_{6,2}) = \DA)$ and
$(\omega(c_{6,2}) = \UA) \rightarrow (\omega(c_{6,3}) =\RA) \rightarrow (\omega(c_{5,5}) =\LA) \rightarrow (\omega(c_{4,5}) =\LA) \rightarrow (\omega(c_{3,3}) =\RA) \rightarrow (\omega(c_{3,2}) = \DA) \rightarrow (\omega(c_{6,2}) = \UA)$.
\end{proof}

\begin{figure}[tb]
\begin{center}
\scalebox{.55}{\begin{tikzpicture}[scale=0.8]
\draw [step =1, black] (-1,-1) grid (11,9);

\draw [line width=1.175mm, black] (0,0) -- (4,0);
\draw [line width=1.175mm, black] (6,0) -- (10,0);
\draw [line width=1.175mm, black] (0,1) -- (10,1);
\draw [line width=1.175mm, black] (0,2) -- (10,2);
\draw [line width=1.175mm, black] (0,3) -- (4,3);
\draw [line width=1.175mm, black] (6,3) -- (10,3);
\draw [line width=1.175mm, black] (2,4) -- (3,4);
\draw [line width=1.175mm, black] (4,4) -- (6,4);
\draw [line width=1.175mm, black] (2,5) -- (3,5);
\draw [line width=1.175mm, black] (4,5) -- (6,5);
\draw [line width=1.175mm, black] (2,6) -- (7,6);
\draw [line width=1.175mm, black] (2,7) -- (5,7);
\draw [line width=1.175mm, black] (6,7) -- (7,7);
\draw [line width=1.175mm, black] (2,8) -- (7,8);
\draw [line width=1.175mm, black] (0,0) -- (0,3);
\draw [line width=1.175mm, black] (2,3) -- (2,7);
\draw [line width=1.175mm, black] (2,7) -- (2,8);
\draw [line width=1.175mm, black] (3,1) -- (3,2);
\draw [line width=1.175mm, black] (3,3) -- (3,7);
\draw [line width=1.175mm, black] (3,7) -- (3,8);
\draw [line width=1.175mm, black] (4,0) -- (4,1);
\draw [line width=1.175mm, black] (4,2) -- (4,5);
\draw [line width=1.175mm, black] (4,6) -- (4,8);
\draw [line width=1.175mm, black] (5,2) -- (5,4);
\draw [line width=1.175mm, black] (5,5) -- (5,8);
\draw [line width=1.175mm, black] (6,0) -- (6,1);
\draw [line width=1.175mm, black] (6,2) -- (6,5);
\draw [line width=1.175mm, black] (6,6) -- (6,8);
\draw [line width=1.175mm, black] (7,1) -- (7,2);
\draw [line width=1.175mm, black] (7,3) -- (7,8);
\draw [line width=1.175mm, black] (10,0) -- (10,3);

\ra{2}{0}
\ra{8}{0}
\ra{2}{1}
\ra{5}{1}
\ra{8}{1}
\ra{2}{2}
\ra{5}{2}
\ra{8}{2}
\ra{2}{4}
\ra{5}{4}
\ra{2}{7}
\ra{3}{7}
\ra{5}{7}
\ra{6}{7}

\la{1}{0}
\la{7}{0}
\la{1}{1}
\la{4}{1}
\la{7}{1}
\la{1}{2}
\la{4}{2}
\la{7}{2}
\la{2}{3}
\la{4}{4}
\la{3}{6}

\ua{3}{5}
\ua{6}{5}
\ua{5}{6}
\ua{6}{6}

\da{4}{3}
\da{5}{3}
\da{3}{4}
\da{6}{4}
\da{2}{5}
\da{2}{6}
\da{4}{6}
\da{4}{7}

\mda{0}{0}
\mua{3}{0}
\mda{6}{0}
\mua{9}{0}
\mda{0}{1}
\mua{3}{1}
\mda{6}{1}
\mua{9}{1}
\mda{0}{2}
\mua{3}{2}
\mda{6}{2}
\mua{9}{2}
\mla{3}{3}
\mla{6}{3}
\mra{4}{5}
\mra{5}{5}
\end{tikzpicture}}\quad\scalebox{.55}{\begin{tikzpicture}[scale=0.8]
\draw [step =1, black] (-1,-1) grid (11,9);

\draw [line width=1.175mm, black] (0,0) -- (4,0);
\draw [line width=1.175mm, black] (6,0) -- (10,0);
\draw [line width=1.175mm, black] (0,1) -- (10,1);
\draw [line width=1.175mm, black] (0,2) -- (10,2);
\draw [line width=1.175mm, black] (0,3) -- (4,3);
\draw [line width=1.175mm, black] (6,3) -- (10,3);
\draw [line width=1.175mm, black] (2,4) -- (3,4);
\draw [line width=1.175mm, black] (4,4) -- (6,4);
\draw [line width=1.175mm, black] (2,5) -- (3,5);
\draw [line width=1.175mm, black] (4,5) -- (6,5);
\draw [line width=1.175mm, black] (2,6) -- (7,6);
\draw [line width=1.175mm, black] (2,7) -- (5,7);
\draw [line width=1.175mm, black] (6,7) -- (7,7);
\draw [line width=1.175mm, black] (2,8) -- (7,8);
\draw [line width=1.175mm, black] (0,0) -- (0,3);
\draw [line width=1.175mm, black] (2,3) -- (2,7);
\draw [line width=1.175mm, black] (2,7) -- (2,8);
\draw [line width=1.175mm, black] (3,1) -- (3,2);
\draw [line width=1.175mm, black] (3,3) -- (3,7);
\draw [line width=1.175mm, black] (3,7) -- (3,8);
\draw [line width=1.175mm, black] (4,0) -- (4,1);
\draw [line width=1.175mm, black] (4,2) -- (4,5);
\draw [line width=1.175mm, black] (4,6) -- (4,8);
\draw [line width=1.175mm, black] (5,2) -- (5,4);
\draw [line width=1.175mm, black] (5,5) -- (5,8);
\draw [line width=1.175mm, black] (6,0) -- (6,1);
\draw [line width=1.175mm, black] (6,2) -- (6,5);
\draw [line width=1.175mm, black] (6,6) -- (6,8);
\draw [line width=1.175mm, black] (7,1) -- (7,2);
\draw [line width=1.175mm, black] (7,3) -- (7,8);
\draw [line width=1.175mm, black] (10,0) -- (10,3);

\ra{2}{0}
\ra{8}{0}
\ra{2}{1}
\ra{5}{1}
\ra{8}{1}
\ra{2}{2}
\ra{5}{2}
\ra{8}{2}
\ra{2}{4}
\ra{5}{4}
\ra{2}{7}
\ra{3}{7}
\ra{5}{7}
\ra{6}{7}

\la{1}{0}
\la{7}{0}
\la{1}{1}
\la{4}{1}
\la{7}{1}
\la{1}{2}
\la{4}{2}
\la{7}{2}
\la{2}{3}
\la{4}{4}
\la{3}{6}

\ua{3}{5}
\ua{6}{5}
\ua{5}{6}
\ua{6}{6}

\da{4}{3}
\da{5}{3}
\da{3}{4}
\da{6}{4}
\da{2}{5}
\da{2}{6}
\da{4}{6}
\da{4}{7}

\mua{0}{0}
\mda{3}{0}
\mua{6}{0}
\mda{9}{0}
\mua{0}{1}
\mda{3}{1}
\mua{6}{1}
\mda{9}{1}
\mua{0}{2}
\mda{3}{2}
\mua{6}{2}
\mda{9}{2}
\mra{3}{3}
\mra{6}{3}
\mla{4}{5}
\mla{5}{5}
\end{tikzpicture}}
\end{center}
\caption{Negated literal-gadget filled in. The filled-in arrows on the left side correspond to a true variable being connected to this gadget, in which case the literal does not lead to satisfying the corresponding clause. When integrated into a clause-gadget (see Figure 6) the upper two filled-in arrows will form a closed cycle with the preset arrows, which leads to the \textsc{Roma} instance not being valid. The filled-in arrows on the right side correspond to a false variable being connected. In this case, the literal leads to the clause being satisfied by breaking this closed cycle.}
\label{fig:nlgf}
\end{figure}

The two possible valid assignments for negative literal-gadgets are shown in \autoref{fig:nlgf}.
Note that in this case $\omega(c_{4,5}) =\RA $ funnels the flow back into the rest of the 1-boxes, while $\omega(c_{4,5}) =\LA$ allows it to leave the gadget downwards, which is a  behavior opposite of positive literal-gadgets.
In \autoref{fig:clg}, it is shown how the literal gadgets are arranged to form a clause gadget. 
Clearly, the gadget can be adapted to contain an arbitrary number of literals.
Whether the literal-gadgets represent positive or negative literals does not matter. W.l.o.g., we chose two positive and one negative literal for the demonstration of the clause-gadget. The individual literal-gadgets are now connected through a cycle of 1-boxes.

\begin{figure}[tb]
\begin{center}
\begin{tikzpicture}[scale=0.33]
\draw [step =1, black] (-1,-1) grid (35,10);

\literal{0}{0}
\literaln{12}{0}
\literal{24}{0}
\draw [line width=0.9mm, black] (0,7) -- (7,7);
\draw [line width=0.9mm, black] (9,7) -- (14,7);
\draw [line width=0.9mm, black] (19,7) -- (31,7);
\draw [line width=0.9mm, black] (33,7) -- (34,7);
\draw [line width=0.9mm, black] (0,8) -- (7,8);
\draw [line width=0.9mm, black] (9,8) -- (14,8);
\draw [line width=0.9mm, black] (19,8) -- (31,8);
\draw [line width=0.9mm, black] (33,8) -- (34,8);
\draw [line width=0.9mm, black] (0,9) -- (34,9);
\draw [line width=0.9mm, black] (0,7) -- (0,9);
\draw [line width=0.9mm, black] (1,7) -- (1,9);
\draw [line width=0.9mm, black] (2,7) -- (2,9);
\draw [line width=0.9mm, black] (3,7) -- (3,9);
\draw [line width=0.9mm, black] (4,7) -- (4,9);
\draw [line width=0.9mm, black] (5,7) -- (5,9);
\draw [line width=0.9mm, black] (6,7) -- (6,9);
\draw [line width=0.9mm, black] (7,8) -- (7,9);
\draw [line width=0.9mm, black] (8,8) -- (8,9);
\draw [line width=0.9mm, black] (9,8) -- (9,9);
\draw [line width=0.9mm, black] (10,7) -- (10,9);
\draw [line width=0.9mm, black] (11,7) -- (11,9);
\draw [line width=0.9mm, black] (12,7) -- (12,9);
\draw [line width=0.9mm, black] (13,7) -- (13,9);
\draw [line width=0.9mm, black] (14,8) -- (14,9);
\draw [line width=0.9mm, black] (15,8) -- (15,9);
\draw [line width=0.9mm, black] (16,8) -- (16,9);
\draw [line width=0.9mm, black] (17,8) -- (17,9);
\draw [line width=0.9mm, black] (18,8) -- (18,9);
\draw [line width=0.9mm, black] (19,8) -- (19,9);
\draw [line width=0.9mm, black] (20,7) -- (20,9);
\draw [line width=0.9mm, black] (21,7) -- (21,9);
\draw [line width=0.9mm, black] (22,7) -- (22,9);
\draw [line width=0.9mm, black] (23,7) -- (23,9);
\draw [line width=0.9mm, black] (24,7) -- (24,9);
\draw [line width=0.9mm, black] (25,7) -- (25,9);
\draw [line width=0.9mm, black] (26,7) -- (26,9);
\draw [line width=0.9mm, black] (27,7) -- (27,9);
\draw [line width=0.9mm, black] (28,7) -- (28,9);
\draw [line width=0.9mm, black] (29,7) -- (29,9);
\draw [line width=0.9mm, black] (30,7) -- (30,9);
\draw [line width=0.9mm, black] (31,8) -- (31,9);
\draw [line width=0.9mm, black] (32,8) -- (32,9);
\draw [line width=0.9mm, black] (33,8) -- (33,9);
\draw [line width=0.9mm, black] (34,7) -- (34,9);

\ra{0}{7}
\ra{1}{7}
\ra{2}{7}
\ra{3}{7}
\ra{4}{7}
\ra{5}{7}
\ra{6}{7}
\ra{9}{7}
\ra{10}{7}
\ra{11}{7}
\ra{12}{7}
\ra{13}{7}
\ra{19}{7}
\ra{20}{7}
\ra{21}{7}
\ra{22}{7}
\ra{23}{7}
\ra{24}{7}
\ra{25}{7}
\ra{26}{7}
\ra{27}{7}
\ra{28}{7}
\ra{29}{7}
\ra{30}{7}

\la{1}{8}
\la{2}{8}
\la{3}{8}
\la{4}{8}
\la{5}{8}
\la{6}{8}
\la{7}{8}
\la{8}{8}
\la{9}{8}
\la{10}{8}
\la{11}{8}
\la{12}{8}
\la{13}{8}
\la{14}{8}
\la{15}{8}
\la{16}{8}
\la{17}{8}
\la{18}{8}
\la{19}{8}
\la{20}{8}
\la{21}{8}
\la{22}{8}
\la{23}{8}
\la{24}{8}
\la{25}{8}
\la{26}{8}
\la{27}{8}
\la{28}{8}
\la{29}{8}
\la{30}{8}
\la{31}{8}
\la{32}{8}
\la{33}{8}

\ua{33}{7}

\da{0}{8}
\end{tikzpicture}
\end{center}
\caption{Clause-gadget consisting of two positive and one negative literal.}
\label{fig:clg}
\end{figure}

\begin{lemma}\label{lem:clause-gadget}
At least one literal-gadget within a clause-gadget needs to allow the flow to leave the gadget downwards, or the assignment will be invalid.
\end{lemma}

\begin{proof}
If the flow is not allowed to leave the gadget downwards through at least one of its literal-gadgets, it is funneled back into the cycle of 1-boxes it stems from. This holds true for  both positive and  negative literal-gadgets. If every literal-gadget within a clause-gadget funnels it back an actual flow-cycle is created, invalidating the overall assignment.
\end{proof}

\begin{figure}[tb]\centering
	\scalebox{1.1}{\begin{tikzpicture}
\filldraw[fill=red!20,draw=none ](0,-0.2) rectangle (7,0.25);
\filldraw[fill=red!20,draw=none ](7,0.25) rectangle (6.55,1.25);
\filldraw[fill=red!20,draw=none ](6.55,1.25) rectangle (2,0.8);
\filldraw[fill=red!20,draw=none ](2,1.25) rectangle (2.45,2.25);
\filldraw[fill=red!20,draw=none ](2.45,2.25) rectangle (4.55,1.8);

\node[text width=0.25cm, anchor=west] at (2,0){$a$};
\node[text width=0.25cm, anchor=west] at (4,0){$b$};
\node[text width=0.25cm, anchor=west] at (6,0){$c$};

\node[text width=0.25cm, anchor=west] at (0,1){$C_{1}$};
\node[text width=0.25cm, anchor=west] at (0,2){$C_{2}$};
\node[text width=0.25cm, anchor=west] at (0,3){$C_{3}$};

\draw[draw=black] (2.25,0.25) -- (2.25,3);
\draw[draw=black] (2.25,3) -- (0.75,3);
\draw[draw=black] (4.25,0.25) -- (4.25,2);
\draw[draw=black] (4.25,2) -- (0.75,2);
\draw[draw=black] (6.25,0.25) -- (6.25,1);
\draw[draw=black] (6.25,1) -- (0.75,1);
\end{tikzpicture}}
	\caption{As depicted in the construction by Lichtenstein~\citen{DBLP:journals/siamcomp/Lichtenstein82} a planar embedding of a \textsc{3SAT} formula is shown where the variables are placed on a horizontal line and the clauses are placed on a vertical line. The variables are then connected with the clauses in which they appear via rectangular lines. If those lines cross, the crossing is replaced by a crossover-gadget, which is depicted in~\autoref{fig:L-crossing}. The core-line connecting the variables with each other is depicted in red.}
	\label{fig:L-coreLine}
\end{figure}
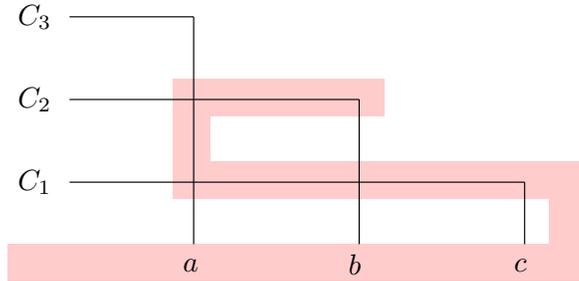

\begin{figure}[tb]\centering
	\scalebox{.94}{\begin{tikzpicture}
\filldraw[fill=red!20 ](0,-0.25) rectangle (12.5,0.25);

\node[text width=0.25cm, anchor=west] at (0,0){$a_{1}$};
\node[text width=0.25cm, anchor=west] at (1.5,0){$\gamma$};
\node[text width=0.25cm, anchor=west] at (3,0){$b_{1}$};
\node[text width=0.25cm, anchor=west] at (4.5,0){$\beta$};
\node[text width=0.25cm, anchor=west] at (6,0){$\xi$};
\node[text width=0.25cm, anchor=west] at (7.5,0){$\delta$};
\node[text width=0.25cm, anchor=west] at (9,0){$b_{2}$};
\node[text width=0.25cm, anchor=west] at (10.5,0){$\alpha$};
\node[text width=0.25cm, anchor=west] at (12,0){$a_{2}$};

\node[text width=0.25cm, anchor=west] at (0.75,1.2){1};
\node[text width=0.25cm, anchor=west] at (1.5,2.2){2};
\node[text width=0.25cm, anchor=west] at (8.25,1.2){3};
\node[text width=0.25cm, anchor=west] at (9.75,1.2){4};
\node[text width=0.25cm, anchor=west] at (9,2.2){5};
\node[text width=0.25cm, anchor=west] at (7.5,3.2){6};
\node[text width=0.25cm, anchor=west] at (7.5,4.2){7};
\node[text width=0.25cm, anchor=west] at (8.25,5.2){8};
\node[text width=0.25cm, anchor=west] at (7.5,6.2){9};

\node[text width=0.25cm, anchor=west] at (11.25,-1.3){10};
\node[text width=0.25cm, anchor=west] at (10.5,-2.3){11};
\node[text width=0.25cm, anchor=west] at (3.75,-1.3){12};
\node[text width=0.25cm, anchor=west] at (2.25,-1.3){13};
\node[text width=0.25cm, anchor=west] at (3,-2.3){14};
\node[text width=0.25cm, anchor=west] at (4.5,-3.3){15};
\node[text width=0.25cm, anchor=west] at (4.5,-4.3){16};
\node[text width=0.25cm, anchor=west] at (3.75,-5.3){17};
\node[text width=0.25cm, anchor=west] at (4.5,-6.3){18};

\draw[draw=black] (0.75,1) rectangle ++(0.5,0.5);
\draw[draw=black] (1.5,2) rectangle ++(0.5,0.5);
\draw[draw=black] (8.25,1) rectangle ++(0.5,0.5);
\draw[draw=black] (9.75,1) rectangle ++(0.5,0.5);
\draw[draw=black] (9,2) rectangle ++(0.5,0.5);
\draw[draw=black] (7.5,3) rectangle ++(0.5,0.5);
\draw[draw=black] (7.5,4) rectangle ++(0.5,0.5);
\draw[draw=black] (8.25,5) rectangle ++(0.5,0.5);
\draw[draw=black] (7.5,6) rectangle ++(0.5,0.5);

\draw[draw=black] (11.25,-1.5) rectangle ++(0.5,0.5);
\draw[draw=black] (10.5,-2.5) rectangle ++(0.5,0.5);
\draw[draw=black] (3.75,-1.5) rectangle ++(0.5,0.5);
\draw[draw=black] (2.25,-1.5) rectangle ++(0.5,0.5);
\draw[draw=black] (3,-2.5) rectangle ++(0.5,0.5);
\draw[draw=black] (4.5,-3.5) rectangle ++(0.5,0.5);
\draw[draw=black] (4.5,-4.5) rectangle ++(0.5,0.5);
\draw[draw=black] (3.75,-5.5) rectangle ++(0.5,0.5);
\draw[draw=black] (4.5,-6.5) rectangle ++(0.5,0.5);

\draw[draw=black] (0.25,0.25) -- (0.25,1);
\draw[draw=black] (0.25,1) -- (0.75,1);
\draw[draw=black] (1.75,0.25) -- (1.75,1);
\draw[draw=black] (1.75,1) -- (0.75,1);

\draw[draw=black] (0.25,0.25) -- (0.25,2);
\draw[draw=black] (0.25,2) -- (1.5,2);
\draw[draw=black] (1.75,0.25) -- (1.75,1.5);
\draw[draw=black] (1.75,1.5) -- (1.75,2);
\draw[draw=black] (3.25,0.25) -- (3.25,2);
\draw[draw=black] (3.25,2) -- (1.5,2);

\draw[draw=black] (7.75,0.25) -- (7.75,1);
\draw[draw=black] (7.75,1) -- (8.25,1);
\draw[draw=black] (9.25,0.25) -- (9.25,1);
\draw[draw=black] (9.25,1) -- (8.25,1);

\draw[draw=black] (9.25,0.25) -- (9.25,1);
\draw[draw=black] (9.25,1) -- (9.75,1);
\draw[draw=black] (10.75,0.25) -- (10.75,1);
\draw[draw=black] (10.75,1) -- (9.75,1);

\draw[draw=black] (7.75,0.25) -- (7.75,2);
\draw[draw=black] (7.75,2) -- (9,2);
\draw[draw=black] (10.75,0.25) -- (10.75,2);
\draw[draw=black] (10.75,2) -- (9,2);

\draw[draw=black] (4.75,0.25) -- (4.75,3);
\draw[draw=black] (4.75,3) -- (7.5,3);
\draw[draw=black] (10.75,0.25) -- (10.75,3);
\draw[draw=black] (10.75,3) -- (7.5,3);
\draw[draw=black] (6.25,0.25) -- (6.25,2.5);
\draw[draw=black] (6.25,2.5) -- (7.75,2.5);
\draw[draw=black] (7.75,2.5) -- (7.75,3);

\draw[draw=black] (4.75,0.25) -- (4.75,4);
\draw[draw=black] (4.75,4) -- (7.5,4);
\draw[draw=black] (10.75,0.25) -- (10.75,4);
\draw[draw=black] (10.75,4) -- (7.5,4);

\draw[draw=black] (4.75,0.25) -- (4.75,5);
\draw[draw=black] (4.75,5) -- (8.25,5);
\draw[draw=black] (12.25,0.25) -- (12.25,5);
\draw[draw=black] (12.25,5) -- (8.25,5);

\draw[draw=black] (3.25,0.25) -- (3.25,6);
\draw[draw=black] (3.25,6) -- (7.5,6);
\draw[draw=black] (4.75,0.25) -- (4.75,5.5);
\draw[draw=black] (4.75,5.5) -- (7.75,5.5);
\draw[draw=black] (7.75,5.5) -- (7.75,6);
\draw[draw=black] (12.25,0.25) -- (12.25,6);
\draw[draw=black] (12.25,6) -- (7.5,6);

\draw[draw=black] (12.25,-0.25) -- (12.25,-1);
\draw[draw=black] (12.25,-1) -- (11.25,-1);
\draw[draw=black] (10.75,-0.25) -- (10.75,-1);
\draw[draw=black] (10.75,-1) -- (11.25,-1);

\draw[draw=black] (12.25,-0.25) -- (12.25,-2);
\draw[draw=black] (12.25,-2) -- (10.5,-2);
\draw[draw=black] (10.75,-0.25) -- (10.75,-1.5);
\draw[draw=black] (10.75,-1.5) -- (10.75,-2);
\draw[draw=black] (9.25,-0.25) -- (9.25,-2);
\draw[draw=black] (9.25,-2) -- (10.5,-2);

\draw[draw=black] (4.75,-0.25) -- (4.75,-1);
\draw[draw=black] (4.75,-1) -- (3.75,-1);
\draw[draw=black] (3.25,-0.25) -- (3.25,-1);
\draw[draw=black] (3.25,-1) -- (3.75,-1);

\draw[draw=black] (3.25,-0.25) -- (3.25,-1);
\draw[draw=black] (3.25,-1) -- (2.25,-1);
\draw[draw=black] (1.75,-0.25) -- (1.75,-1);
\draw[draw=black] (1.75,-1) -- (2.25,-1);

\draw[draw=black] (4.75,-0.25) -- (4.75,-2);
\draw[draw=black] (4.75,-2) -- (3,-2);
\draw[draw=black] (1.75,-0.25) -- (1.75,-2);
\draw[draw=black] (1.75,-2) -- (3,-2);

\draw[draw=black] (7.75,-0.25) -- (7.75,-3);
\draw[draw=black] (7.75,-3) -- (4.5,-3);
\draw[draw=black] (1.75,-0.25) -- (1.75,-3);
\draw[draw=black] (1.75,-3) -- (4.5,-3);
\draw[draw=black] (6.25,-0.25) -- (6.25,-2.5);
\draw[draw=black] (6.25,-2.5) -- (4.75,-2.5);
\draw[draw=black] (4.75,-2.5) -- (4.75,-3);

\draw[draw=black] (7.75,-0.25) -- (7.75,-4);
\draw[draw=black] (7.75,-4) -- (4.5,-4);
\draw[draw=black] (1.75,-0.25) -- (1.75,-4);
\draw[draw=black] (1.75,-4) -- (4.5,-4);

\draw[draw=black] (7.75,-0.25) -- (7.75,-5);
\draw[draw=black] (7.75,-5) -- (3.75,-5);
\draw[draw=black] (0.25,-0.25) -- (0.25,-5);
\draw[draw=black] (0.25,-5) -- (3.75,-5);

\draw[draw=black] (9.25,-0.25) -- (9.25,-6);
\draw[draw=black] (9.25,-6) -- (4.5,-6);
\draw[draw=black] (7.75,-0.25) -- (7.75,-5.5);
\draw[draw=black] (7.75,-5.5) -- (4.75,-5.5);
\draw[draw=black] (4.75,-5.5) -- (4.75,-6);
\draw[draw=black] (0.25,-0.25) -- (0.25,-6);
\draw[draw=black] (0.25,-6) -- (4.5,-6);
\end{tikzpicture}}
	\caption{Translation of the crossover-gadget replacing crossings presented in \cite{DBLP:journals/siamcomp/Lichtenstein82} for the setting of \textsc{Roma}. We connect variables $b$ and $b_{2}$ as well as $a$ and $a_2$, respectively, by utilizing straight-line gadgets, which saves us the trouble of creating clauses specifically to copy their assignments. The specific formula reads, as presented in \cite{DBLP:journals/siamcomp/Lichtenstein82}, as follows:
	$
	(\neg a_{1} \vee \neg\gamma) 
	\wedge 
	(a_{1} \vee b_{1} \vee \gamma) 
	\wedge 
	(b_{2} \vee \neg\delta) 
	\wedge
	(b_{2} \vee \neg\alpha) 
	\wedge
	(\neg\delta \vee \neg\alpha)	
	\wedge 
	(\alpha \vee \beta \vee \xi) 
	\wedge 
	(\neg\alpha \vee \neg\beta) 
	\wedge 
	(a_{2} \vee \neg\beta)
	\wedge 
	(\neg a_{2} \vee b_{1} \vee \beta) 
	\wedge 
	(a_{2} \vee \neg\alpha) 
	\wedge 
	(\neg a_{2} \vee \neg b_{2} \vee \alpha) 
	\wedge 
	(\neg b_{1} \vee \neg\beta) 
	\wedge
	(\neg b_{1} \vee \neg\gamma) 
	\wedge	
	(\neg\beta \vee \neg\gamma) 
	\wedge 
	(\gamma \vee \delta \vee \neg\xi) 
	\wedge
	(\neg\gamma \vee \neg\delta) 
	\wedge 
	(\neg a_{1} \vee \neg\delta) 
	\wedge 
	(a_{1} \vee \neg b_{2} \vee \delta) 
	$.}
	\label{fig:L-crossing}
\end{figure}
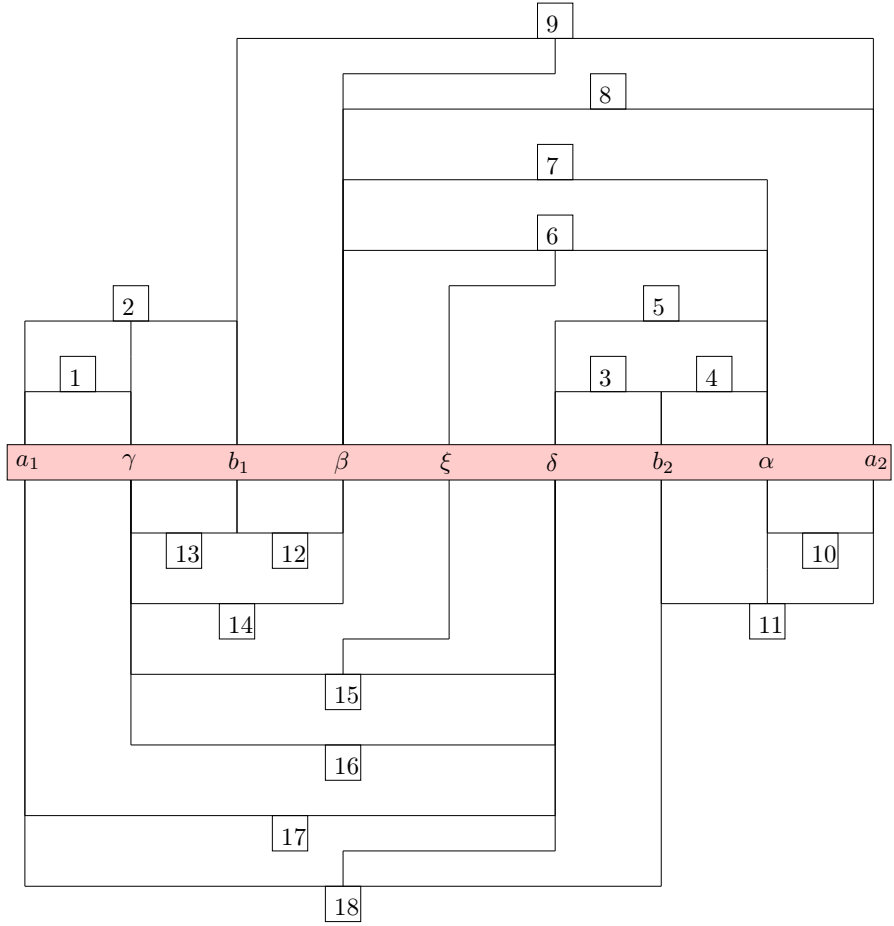

\begin{figure}[tb]\centering
	\scalebox{0.89}{\begin{tikzpicture}[scale=0.2]
\filldraw[fill=red!20 ](4,42) rectangle (55,43);

\variable{5}{36}
\variable{23}{36}
\variable{41}{36}

\node[text width=0.25cm, anchor=west] at (10,40){$a_{1}$};
\node[text width=0.25cm, anchor=west] at (28,40){$\gamma$};
\node[text width=0.25cm, anchor=west] at (46,40){$b_{1}$};

\node[text width=5cm, anchor=west] at (24,67){Clause 1};
\node[text width=5cm, anchor=west] at (24,97){Clause 2};
\node[text width=5cm, anchor=west] at (45,16){Clause 13};

\clause{11}{62}
\clause{11}{92}

\edge{5}{47}
\edge{5}{50}
\edge{5}{53}
\edge{5}{56}
\fanout{5}{59}
\edge{5}{62}
\edge{5}{65}
\edge{5}{68}
\edge{5}{71}
\edge{5}{74}
\edge{5}{77}
\edge{5}{80}
\edge{5}{83}
\edge{5}{86}
\fanout{5}{89}

\edge{29}{47}
\edge{29}{50}
\edge{29}{53}
\edge{29}{56}
\fanout{29}{59}
\fanout{35}{59}
\fanout{41}{59}
\edge{47}{62}
\edge{47}{65}
\edge{47}{68}
\edge{47}{71}
\edge{47}{74}
\fanout{41}{74}
\fanout{35}{74}
\fanout{29}{74}
\edge{29}{77}
\edge{29}{80}
\edge{29}{83}
\edge{29}{86}
\edge{29}{89}

\fanout{47}{47}
\edge{53}{50}
\edge{53}{53}
\edge{53}{56}
\edge{53}{59}
\edge{53}{62}
\edge{53}{65}
\edge{53}{68}
\edge{53}{71}
\edge{53}{74}
\edge{53}{77}
\edge{53}{80}
\edge{53}{83}
\edge{53}{86}
\fanout{47}{89}
\fanout{41}{89}
\edge{53}{92}
\edge{53}{95}
\edge{53}{98}


\clauser{32}{21}

\edge{8}{12}
\edge{8}{15}
\edge{8}{18}
\edge{8}{21}
\edge{8}{24}
\edge{8}{27}
\edge{8}{30}
\edge{8}{33}

\edge{26}{12}
\edge{26}{15}
\edge{26}{18}
\fanout{26}{21}
\edge{26}{24}
\edge{26}{27}
\edge{26}{30}
\edge{26}{33}

\edge{50}{33}
\edge{50}{30}
\edge{50}{27}
\edge{50}{24}
\fanout{50}{21}
\fanout{56}{21}
\fanout{62}{21}
\edge{68}{18}
\edge{68}{15}
\edge{68}{12}
\end{tikzpicture}}
	\caption{Here we show a detailed \textsc{Roma} construction for a part of the gadget shown above. Again, the core-line is depicted in red. We connect variables $a_{1}, \gamma$ and $a_{2}$ to clauses 1, 2 and 13.}
	\label{fig:L-Roma}
\end{figure}

Since the lower part of the literal-gadget has the exact same structure as the fanout-gadget, this part serves as a natural place to connect conductorboxes which are part of fanout- or straightline-gadgets. Through this connection we can propagate the assignment of a connected variable-gadget. Note that the part of the literal-gadget which resembles a fanout-gadget cannot be replaced by a structure resembling a simple straight-line gadget as this would destroy the property of being a parsimonious reduction. Due to the ability of the fanout-gadget to propagate the signal of one variable-gadget to multiple literal-gadgets, we can check for each literal-gagdet, even for those belonging to different clause-gadgets, whether the assignment of the connected variable-gadget allows the flow of the literal-gadgets, and thus the clause-gadgets they are a part of, to reach the Roma-cell and thus satisfy the clause. Even though each literal-gadget ends up with two points which can connect to conductor boxes, it does not matter whether the same variable is connected to both or just one of the same, as the signal will be propagated either way. Two different variable-gadgets will not be connected to the same literal-gadget.

How exactly the individual gadgets are connected is depicted in~\autoref{fig:L-Roma}, where a small part of the crossover-gadget in Lichtenstein's construction is realized as a Roma board, e.g., variable-gadget $\gamma$ is connected to literals from both the gadgets of Clause~1 and Clause~2 utilizing straightline- and fanout-gadgets. 

Finally, all variable gadgets will be connected via the core-line (as described above), which will make it possible that all paths lead to the Roma-cell, which is placed on the leftmost cell of the core-line. Only a single cell within a variable-gadget is assigned in order to represent the assignment of the corresponding variable. The resulting signal will be passed on to the connected literal-gadgets in an unambiguous manner---after that first assignment, there is only one valid assignment for the entire network consisting of a variable-gadget and straight-line, fanout-, and literal-gadgets. Overall, there are only two valid assignments for this network. The correctness of the overall construction follows by the lemmas stated so far.
Hence, the original 3-SAT formula $\varphi$ has a satisfying assignment if and only if the constructed \textsc{Roma} instance $R(\varphi)$ has a solution.

As can be seen in~\autoref{fig:L-Roma}, there are still undefined cells between the connected gadgets. Each of these cells will be prefilled and forms their own box. At this point, it is important to realize that we can construct paths from every single one of these cells either directly to the core-line, which leads to the Roma-cell, or, in case of an encapsulated space, like the one between Clause~1 and Clause~2, to the upper part of a clause-gadget.
If we connect cells to a clause-gadget like that, the Roma-cell can be reached from all of these cells if the corresponding clause is satisfied in the same fashion as it can be reached from the cells which are part of the gadget itself.
Clearly, the very details of the described filling of these boxes do not matter, it is only important 
that all paths lead to Rome.


\section{Further Complexity-Theoretic Consequences}

In this section, we examine our reduction more carefully, so that we can deduce several further interesting consequences.
The first one deals with potential limitations of exponential-time algorithms designed to solve a \textsc{Roma} instance, as provided by the Exponential-Time Hypothesis, or ETH for short; see \cite{DBLP:journals/jcss/ImpagliazzoP01}. This is interesting, as we will provide in the next section some matching algorithmic results.

\begin{theorem}\label{thm:ETH-Roma}
Assuming ETH, there is no $\mathcal{O}\left(2^{{o}(n)}\right)$-algorithm 
for 
solving $n\times n$-\textsc{Roma}~puzzles.
\end{theorem}

\begin{proof}
Here, we have to dig a bit deeper into the \NP-hardness reduction of Lichtenstein \cite{DBLP:journals/siamcomp/Lichtenstein82}.
He presents a specific design that aligns all variables along the $x$-axis and all clauses along the $y$-axis and then connects the variable to the 
clauses that they are contained in by axes-parallel lines. Of course, these straight lines will intersect, but he explains how to introduce crossover-gadgets that will introduce a constant number of new variables and new clauses to replace the crossings. We can first apply the famous sparsification lemma of Impagliazzo and Paturi \cite{DBLP:journals/jcss/ImpagliazzoP01} to guarantee that the number of variables and the number of clauses in the given 3-SAT formula are of the same order, say, $N$. 
As there are at most $N^2$ crossings in the rectangular drawing, there will be no more than $\mathcal{O}(N^2)$ many variables and clauses in the instance of \textsc{Planar 3SAT} that Lichtenstein proposes. Because each crossover-gadget can be simulated by a network in a \textsc{Roma} instance that uses area $\mathcal{O}(1)$ only, see the Figures~\ref{fig:L-coreLine}, \ref{fig:L-crossing}, and \ref{fig:L-Roma}, we can build a \textsc{Roma} instance to the given \textsc{3SAT} instance with at most $N$ variables and at most $N$ clauses that is of size $\mathcal{O}(N^2)$. Hence, if there was an algorithm that would solve    an $n\times n$-\textsc{Roma} puzzle in time $\mathcal{O}\left(2^{{o}(n)}\right)$, then one could solve any  \textsc{3SAT} instance with at most $N$ variables and at most $N$ clauses in time $\mathcal{O}\left(2^{{o}(N)}\right)$, contradicting ETH. 
\end{proof}

We now turn to a counting variant of our main combinatorial problem: 
\#\textsc{Roma} is the problem  to count the number of solutions of a Roma puzzle.
We might not like the idea that a given puzzle has too many solutions, as it seems to be the case then that
a human might find it awkward to play the game, as seemingly not much cleverness is needed.
This is why these puzzles are often designed in a way that they admit only a unique solution.
Still, counting the number of solutions of a Roma puzzle is quite infeasible.

\begin{theorem}\label{thm:CountRoma}
$\#$\mdseries\textsc{Roma} is $\#$\Ptime-complete.
\end{theorem}

\begin{proof}
As it is easy to design a Turing machine that first guesses an assignment (a polynomial-size witness suffices by \NP-membership) and then deterministically 
verifies the validity of it, we can also design a nondeterministic polynomial-time Turing machine that has as many accepting paths as the original
\textsc{Roma} instance has solutions.
Conversely, as the counting variant \textsc{Planar $\#$3SAT} is $\#$\Ptime-complete~\cite{doi:10.1137/S0097539793304601} and the reduction from \textsc{Planar 3SAT} is parsimonious (this can be checked by carefully going through the proofs of the previous lemmas again), the number of satisfying assignments of this 3-SAT formula equals the number of solutions of the constructed \textsc{Roma} instance. 
\end{proof}

Finally, for the design of uniquely solvable Roma puzzles, it would be good to know how many hints have to be added (as preset arguments) to make a puzzle uniquely solvable. We call this problem \textsc{FCP \textsc{Roma}}, to be spelled out as \textsc{Fewest Clues Problem}. More formally, it is asked if there exist at most $k$ empty cells that can be assigned in such a manner that the remaining instance  only has a single valid assignment, given a \textsc{Roma} instance and an integer~$k$ as inputs.
This relates to \textsc{FCP \textsc{Planar 3SAT}}. Here, we are given a planar \textsc{3SAT} instance and an integer~$k$, and it is asked if there exists a partial assignment of at most $k$ variables such that the remaining formula  has only one satisfying assignment.

\begin{theorem} \label{thm:FCPRoma}
\textsc{FCP \textsc{Roma}} is $\Sigma^{P}_{2}$-complete.
\end{theorem}

\begin{proof}
	We constructed an instance of \textsc{Roma} which is equivalent to a given instance of \textsc{Planar 3SAT} above. Instead of a partial assignment of $k$ variables, we now assign $k$ empty cells. For each assignment of a variable in the instance of \textsc{Planar 3SAT} reduced from, we need to assign exactly one cell in the resulting instance of \textsc{Roma}, specifically one cell of the corresponding variable-gadget as described above. Thus, the number of additional variables to be assigned $k$ is preserved. As \textsc{FCP \textsc{Planar 3SAT}} is $\Sigma^{P}_{2}$-complete (see 
\cite{DBLP:journals/tcs/DemaineMSWA18}), \textsc{FCP \textsc{Roma}} is $\Sigma^{P}_{2}$-hard. Conversely, in  Section 3.1 of \cite{DBLP:journals/tcs/DemaineMSWA18} it is shown that the fewest clues problem of every  problem from \NP is in $\Sigma^{P}_{2}$.  
\end{proof}

After exhibiting these complexity-theoretic limitations of our problem, it is interesting to see algorithms that can possibly meet the limitations. This is the theme of the next section, where we first explain a rather simple branching algorithm and then sketch a more sophisticated dynamic programming algorithm.

\section{Algorithms for Solving \bfseries{\scshape{Roma}}}

\subsection*{A Search Tree Algorithm for \mdseries\textsc{Roma}}

Let $n^2$ be the total number of cells in a given instance of \textsc{Roma} $\mathcal{R}$. Let $k$ be the number of empty cells in $\mathcal{R}$. 
A first approach would test all possible four assignments of each empty cell.
This will, at worst, result in a running time of $\mathcal{O}(4^{k} \cdot n^2)$, because we have to check all $n^2$ cells to decide if the assignment is valid or not for all $4^{k}$ possible assignments. Naturally, this will be faster the fewer empty cells $\mathcal{R}$ has. Overall the algorithm is polynomial in $n$ and exponential in $k$. This shows \textsc{Roma} to be fixed parameter tractable in the standard parameter~$k$. But can we do better? This will be examined in the following.

Consider any  $c_{ij}\in\mathcal{E}_{\mathcal{R}}$. From a naive point of view, there are 4 possible assignments for $c_{ij}$: $\LA, \RA, \DA$ and $\UA$. However, in many cases we can narrow it down to a single assignment in polynomial time. Consider the following:
We start with all 4 possible assignments in mind. We then check the other cells of box $b_{m}$, which contains $c_{ij}$. Each assignment already given within that box is no longer an option for $c_{ij}$. Next, we check each true neighbor of $c_{ij}$. We follow the flow of each of these cells. Should $c_{ij}$ pointing towards them lead to a closed cycle (consisting of 2 cells or more) this direction is not an option for a valid assignment. This includes borders of the board, since arrows cannot point out of the instance according to the rules of Roma. If only one valid assignment is left after these checks, we assign $c_{ij}$, creating  a new instance~$\mathcal{R'}$ in the process, and recurse our procedure with~$\mathcal{R'}$ as input.
In the worst case, no savings are possible this way if we only consider boxes of size one.
However, we get better estimates in other cases.

\begin{lemma}\label{lem:FPT}
If $\mathcal{R}$ is an $n\times n$-\textsc{Roma} puzzle with $k$ empty cells, but without empty cells that form 1-boxes, then $\mathcal{R}$
can be solved in time  $\mathcal{O}(3.32^{k} \cdot n^2)$.
\end{lemma}

\begin{proof}
The worst case clearly happens with 2-boxes now. Let $b=\{c_{i,j},c_{i',j'}\}$ be a 2-box.
If one of the two cells is pre-set, then we have at most three possibilities for the empty cell.
Otherwise, we find four possibilities for setting $\omega(c_{i,j})$. In the three cases where $\omega(c_{i,j})$ does not point to $c_{i',j'}$,
we have each time three possibilities for setting $\omega(c_{i',j'})$. In the case where $\omega(c_{i,j})$ points towards  $c_{i',j'}$,
only two possibilities remain for  $\omega(c_{i',j'})$. Hence, altogether we have (at most) 11 possibilities to set $\omega$ on $b$, which is $\sqrt{11}\leq 3.32$ per cell.
\end{proof}

Further improvements are possible if there are  only few 2-boxes and no 1-boxes, but we refrain from giving further details here, because still, $k$ is of the order of $n^2$ if we assume that relatively few hints are given at the beginning. Therefore, the approach presented in the next subsection is (at least theoretically) more interesting.

\subsection*{A Dynamic-Programming Algorithm for \mdseries\textsc{Roma}}

We can get an algorithm that matches the lower bound of \autoref{thm:ETH-Roma}.

\begin{theorem}\label{thm:DPalgo}
There is an $\mathcal{O}\left(2^{\mathcal{O}(n)}\right)$-algorithm 
solving an $n\times n$-\textsc{Roma} puzzle.
\end{theorem}

The algorithm is based on a dynamic programming (DP) approach along the rows of a \textsc{Roma} puzzle board. The main difficulty in obtaining the claimed run-time bounds consists in the problem to check the graph condition of acyclicity. A naive approach would end up with a run-time of $\mathcal{O}\left(2^{\mathcal{O}(n^2)}\right)$, as a natural idea would be to memorize for every pair of cells on the `sweep row' (as we will call the current row) whether or not a path leads from
one cell to the other through the already processed area of the board. To overcome this difficulty, we make use of Catalan structures, similarly as proposed in  \cite{DBLP:journals/iandc/BodlaenderCKN15,DBLP:journals/jcss/DornFT12} for quite different problems on planar graphs that deal with connectivity constraints.
As the name suggests, these structures are related to a proper bracketing that models the paths finally leading to the Roma-cell. We develop a special syntax for these bracket structures to reflect their meaning with respect to configurations of the Roma game.

\begin{proof}
The basic idea of the algorithm is to use dynamic programming (DP) along the rows of a \textsc{Roma} puzzle board. This sliding row works similar as a sweep line in computational geometry, steadily moving downwards. We will therefore speak about the (current) \emph{sweep row} and the \emph{successor row}. 
To each row with $n$ squares, we associate a string of length at most $3n$ over the alphabet $\Sigma=\Delta\cup\Delta\times\Delta\cup \Delta\times\binom{\Delta}{2}\cup B$,
where $\Delta=\{\mbox{$\circ$},\UA,\DA,\RA,\LA\}$ is the alphabet of Roma cell states, and $B=\{\LB,\RB\}$ are brackets.
We can formulate further restrictions that would reduce the number of configurations considerably, but we refrain from giving these details here, as they are immaterial to the claim.
Let us explain the meaning of such a word encoding a \emph{row configuration} by an example for $n=6$:
\begin{equation}\label{sweeprow1}
\scalebox{.94}{$\DA\LB\LB\LA\LB\LP\UA,\{\LA,\DA\}\RP\RB\LP\DA,\{\RA,\LA\}\RP\UA\RB\UA\RB$}\ \text{refers to the row}\ \text{
\begin{tikzpicture}[scale=0.34]\draw [step =1, black]grid (6,1);\da{0}{0}\la{1}{0}\ua{2}{0}\da{3}{0}\ua{4}{0}\ua{5}{0}; \draw(0.5,1) .. controls (2.6,3) and (4,3) .. (5.5,1);
\draw(0.5,1) .. controls (3,2.7) and (3.7,2.5) .. (4.5,1);
\draw(2.5,1) .. controls (2.6,2) and (3.5,2) .. (3.5,1);
\end{tikzpicture}
}\end{equation}
but encodes much more information. 
It also tells the box information that is not shown in this picture, but that is important to keep from the previous rows, stored in an abstract fashion. This information is necessary to compute all configurations of the successor row.
\begin{itemize}
\item A symbol from $\Delta$ means, in the first place, that this symbol is sitting in that cell of the sweep row. Besides this, it can encode two different things: either the box of this cell is not continued in the successor row (in which case we say that the symbol is of type~0), or this is the only symbol already fixed for the box (i.e., this box was started on the sweep row), in which case we say that the symbol is of type~3. Which of the two cases occurs can be decided by the DP algorithm by checking the given puzzle board.
\item A symbol  $(a,b)\in\Delta\times\Delta$ says: (1) Symbol~$a$ is in that cell of the sweep row. (2) Because the box of this cell continues in the successor row, $b$ is the only symbol that is still available in that box. We also say that the combined symbol $(a,b)$ is of type~1.
\item A symbol $(a,\{b,c\}) \in \Delta\times\binom{\Delta}{2}$ means: (1) Symbol~$a$ is in that cell of the sweep row. (2) The box of this cell continues in the successor row and possibly beyond, and $b,c$ are the symbols that are still available in that box. We also say that the combined symbol $(a,\{b,c\}) \in \Delta\times\binom{\Delta}{2}$ is of type~2.
\end{itemize}
Using formal language terminology, from a row configuration $w\in\Sigma^*$, the \emph{row content} $w'\in\Delta^*$ can be retrieved by a morphism $h$ which maps
brackets to the empty word, projects $(a,x)\mapsto a$ for $(a,x)\in \Delta\times\Delta\cup \Delta\times\binom{\Delta}{2}$ and works as the identity on~$\Delta$; see our example.
A type that is bigger than zero tells the number of possibilities that are still available for the remaining cells of that box. In a sense, type~0 is an exception, as obviously no information has to be transferred into a ``fresh row''.

To understand how updates in the DP procedure could work, we continue with our example, see \autoref{fig:sweeprow2}.
In order to describe how the next row could be formed, we have to display a bit more of the board. The light blue first row is showing only one possibility of how the previous row (with respect to the sweep row in dark blue) could look like, not everything is enforced or known at this step. In the last row, a possible successor row is indicated in gray, although the two first gray arrows are enforced by the sweep row. The black arrows were given as hints in the very beginning.
\begin{figure}[H]
\begin{center}
\mbox{\hspace{-7em}}\ 
\begin{minipage}[b]{9em}
\begin{tikzpicture}[scale=0.5]\draw [step =1, black]grid (6,3); \cra{0}{2}\da{1}{2}\cra{2}{2}\cda{3}{2}\cua{4}{2}\cla{5}{2}\bda{0}{1}\bla{1}{1}\bua{2}{1}\bda{3}{1}\bua{4}{1}\bua{5}{1}\gda{0}{0}\la{1}{0}\gda{2}{0}\gra{3}{0}\gua{4}{0}\gla{5}{0};
 \draw(1.5,3) .. controls (2.6,5) and (3.4,5.1) .. (4.5,3);
\draw(3.5,3)--(3.5,4);
\draw [line width=0.7mm, black] (0,0) -- (0,3);
\draw [line width=0.7mm, black] (1,2) -- (1,3);
\draw [line width=0.7mm, black] (1,0) -- (1,1);
\draw [line width=0.7mm, black] (2,1) -- (2,3);
\draw [line width=0.7mm, black] (3,0) -- (3,3);
\draw [line width=0.7mm, black] (5,1) -- (5,3);
\draw [line width=0.7mm, black] (4,0) -- (4,1);
\draw [line width=0.7mm, black] (6,0) -- (6,3);
\draw [line width=0.7mm, black] (0,1) -- (2,1);
\draw [line width=0.7mm, black] (1,3) -- (5,3);
\draw [line width=0.7mm, black] (1,2) -- (2,2);
\draw [line width=0.7mm, black] (3,2) -- (5,2);
\draw [line width=0.7mm, black] (1,0) -- (3,0);
\draw [line width=0.7mm, black] (1,3) -- (3,3);
\draw [line width=0.7mm, black] (4,1) -- (6,1);
\end{tikzpicture}\end{minipage}
\quad
\begin{minipage}[b]{4em}
With possible\\ bracket structure
\end{minipage}
\quad
\begin{minipage}[b]{12em}
$\LP\RA,\{\DA,\LA\}\RP\DA\LB\RA\DA\UA\RB\LP\LA,\UA\RP$\\[1ex]
$\DA\LB\LB\LA\LB\LP\UA,\{\LA,\DA\}\RP\RB\LP\DA,\{\RA,\LA\}\RP\UA\RB\UA\RB$\\[1ex]
$\DA\LB\LA\DA\LP\RA,\LA\RP\LP\UA,\{\DA,\RA\}\RP\RB\LP\LA,\{\DA,\RA\}\RP$
\end{minipage}
\end{center}
\caption{The sweep row inherits information from its predecessor and passes it to its successor.}
\label{fig:sweeprow2}
\end{figure}
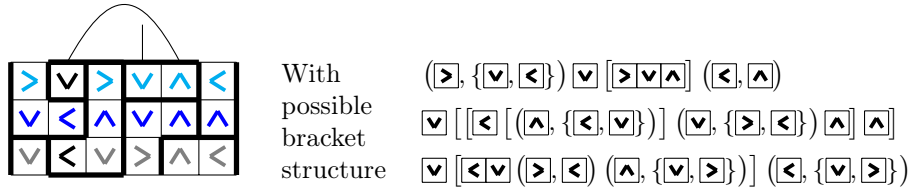
How does a \emph{successor row arrow consistency check} work? We are considering the string $\DA\LA\DA\RA\UA\LA\in D^6$ as a possible successor, one of many that we have to check.
The first $\DA$ is possible, as this starts a new box. Note that this is the only choice for this cell, as $\RA$ would contradict with the pre-set $\LA$ in the second cell, and further neither of $\LA$ nor $\UA$ are possible, as $\LA$ is pointing to the left wall and $\UA$ would contradict $\DA$ on the sweep row above. The next symbol $\LA$ is given in the beginning. The symbol above on the sweep row does not contradict with this preset hint. As the box is finished in this row, the symbol has no second component.
The third symbol is $\DA$; it sees $\LP\UA,\{\LA,\DA\}\RP$ above, which means that in the current box, the only symbols are $\LA$ and $\DA$ that could be put into the cell currently considered. But as left to it $\LA$ was fixed, the symbol $\DA$ is indeed enforced. As the fourth symbol, we chose $\RA$. Again, the box was already considered in the sweep row, and the symbol above was $\LP\DA,\{\RA,\LA\}\RP$, so that $\RA$ is one of the two possibilities. However, as the box is further continued beyond the successor row, we have to introduce the combined symbol $\LP\RA,\LA\RP$. The last two symbols, $\UA$ and $\LA$, are both in the same (new) box and hence not conflicting the sweep row above. However, they get a second component $\{\DA,\RA\}$ to propagate the possibilities for the 4-box that will be completed in the next row. The reader may check that also the predecessor row of the sweep row is arrow-consistent with the sweep row.

We are now describing the role of the bracket structure that has to be used in the \emph{successor row cycle consistency check} explained below. This is important as we have to prevent directed cycles in a constructed solution. 
We are first assuming that the Roma-cell is not in the upper (`forgotten') part of the board.
This means that every path that enters the upper zone via a symbol $\UA$ in the sweep row must leave the upper zone again via a symbol $\DA$ in the sweep row.
How exactly this path goes through the upper part is not important. However, one can imagine these paths as forming a kind of river system. Planarity ensures that these paths never cross, but they could merge. Continuing with the analogy, one could try to draw some river basins. This is how one could interpret the drawings of the lines in the pictures. In the sweep row picture in \eqref{sweeprow1}, the rightmost upward arrow starts a path that goes all way along down again to the first downward arrow. The penultimate upward arrow again starts a river that flows (to stay in the picture) to the North-West, turning South again to leave the area again via the first downward arrow. However, the first upward arrow starts a river that flows to the North-East before turning South again to leave the upward area via the second downward arrow.
The role of the brackets is to describe these river basins in a unique way. In our example, the first path that we described corresponds to the outermost matching pair of brackets. How is this constructed? The upward arrow points North-West, and this lets us insert a closing bracket to the right of the upward arrow symbol. The matching opening bracket is inserted to the right of the downward arrow that indicates where this path leaves the upper part again. To formulate this bracket setting rule more explicitly: Brackets are set to include the upward arrows (where the river starts) but excludes the downward arrow (where the mouth of the river is). The reason behind this convention of excluding the downward arrows from the brackets is that there might be rivers that start in the upper part but share a mouth with a river that actually starts at the sweep row. 
Furthermore, downward arrows could receive paths from two directions, and it would destroy the meaning of the bracket structure if we would have an opening bracket to the left of a downward arrow and a closing bracket to its right.
The second path starts at the penultimate upward arrow, so we insert a closing bracket  to the right of the upward arrow symbol. We insert the matching opening bracket to the right of the downward arrow that indicates where this path leaves the upper part again. Finally, the innermost pair of brackets indicates a flow from the first upward arrow to the North-East. Therefore, an opening bracket is inserted to the left of this upward arrow, and we insert a matching closing bracket to the left of the downward arrow.
Our conventions also imply that for each upward arrow, we have a pair of brackets, either with the opening bracket sitting to the left of the upward arrow, or with the closing bracket sitting to the right of the upward arrow,
while for the downward arrows, it could be the case that there is no bracket attached to it (which means that these rivers do not start at the sweep row), or that even several brackets are attached to it. Notice that closing brackets are always sitting immediately to the left of a downward arrow, while opening brackets are sitting to the right of a downward arrow.

Let us add one further thought about the brackets: If we have, say, two paths that start at the beginning of a row and move somehow through the upper part of the board, to come down in two different downward arrows, then it cannot be the case that the first path ends at the penultimate downward arrow, while the second path ends at the last downward arrow, because this would mean that these two paths have crossed, which is impossible as the overall structure is planar. Similar ideas have been exploited implicitly in \cite{DBLP:journals/iandc/BodlaenderCKN15,DBLP:journals/jcss/DornFT12}; the so-called Catalan structures are referring to our explicit bracketings which catch these ideas in a transparent way. We will exploit these connections again when counting the possible configurations below.

The  \emph{successor row cycle consistency check} will do two main things. Obviously, a successor row has to be rejected if it closes a path through the upper part.
Moreover, any bracket structure in a sweep row that is otherwise consistent with the currently considered successor row induces a bracket structure on the successor row in a unique manner. Let us describe this last point with our example again. In the board picture of \autoref{fig:sweeprow2}, we see the bracket structure and a sketch of the corresponding river basin of the predecessor row. There is only one path starting at the only upward arrow of that row, going North-West and entering that row again at a pre-set downward arrow.
This is supposed to be the path information that was propagated to the sweep row in the following manner: For each upward arrow of the sweep row, we checked where it leads to in the previous row. For instance, the last upward arrow points to a left arrow; following this further (at the time when we constructed this configuration of what is now the sweep row, the predecessor row was known to full extent), we encounter an upward arrow whose matching downward arrow happens to be pre-set. This path moves down to the sweep row and on the sweep row, the left arrow moves to another downward arrow. Hence, we can draw this path (as shown in the sweep row picture in~\eqref{sweeprow1}) from the rightmost to the leftmost cell of the sweep row. The other connections are determined in an analogous fashion. The reader is invited to also check the bracket structure of the successor row, which again consists of a single pair of brackets only.

This description should suffice to explain how to formally prove the claim by induction.

However, notice that there could be several sweep row configurations that are compatible with one string describing the cell contents of the successor row. Therefore, there could be several different bracket structures that can be associated to such a string over the alphabet~$\Delta$. 

In order to count the possible configurations, we will go step-by-step. 

First, assume that we consider a fixed bracket structure, and we also fixed for each bracket pair which of the two brackets (opening or closing) is associated to an upward arrow. In case we have combined symbols, we only consider whether the first component contains an upward arrow.
Assume we have $n^{[\,]}$ bracket pairs involved (as we see this many upward arrows on the sweep row) and that we have $n^{[\,]}_{\scalebox{.5}{\DA}}$ downward arrows that could be matched with these brackets on the sweep row.

As the bracket structure is fixed, there are no more than $\binom{ n^{[\,]}+n^{[\,]}_{\scalebox{.5}{\DA}}}{n^{[\,]}_{\scalebox{.5}{\DA}}}$ many possibilities for such a matching.
(In fact, there are less, as the bracket structure is neglected in this type of counting, but we are not optimizing the counting here.) 

Next, assume that we have $n^{[\,]}_{\scalebox{.5}{\UA}}$ upward arrows that are involved in potential cycles. Recall that if the Roma-cell is in the upper part of the board, then not all upward arrows need to be in a path that returns to the sweep row. As each such upward arrow gives rise to a bracket pair, we only have to decide if the opening or closing bracket of this bracket pair is associated to such an upward arrow. It is well-known that the number of ways we can properly form bracket expressions with $p$ pairs of brackets is given by $\frac1{p+1}\binom{2p}{p}$, which is upper-bounded by $4^p$  (this is also known as Catalan numbers; see, e.g., Chapter 14 in \cite{LinWil2006}). Therefore, we have $8^{n^{[\,]}_{\scalebox{.5}{\UA}}}$ many possibilities to create bracket structures and fix the association to upward arrows. 

Next, we have to reason a bit about the type of the letters from $\Delta$ and how to count them. As the Roma-cell is unique and pre-set, this does not enter the following considerations.
Recall that the type of a letter is completely determined by the  given  box structure. Therefore, when we consider the set of configurations of the sweep row, each cell may host different letters, but all of them have the same type. Clearly, there are only 4 different letters of type~0 and of type~3. There are 12 different combined letters of type~1, because the combined letters $(a,a)\in\Delta\times\Delta$ would never appear, as it makes no sense to put the symbol $a$ into a cell and tell, at the same time, that the last symbol that could be set in that box is also~$a$. Similarly, there are also 12 different combined letters of type~2, because when we put $a$ into the cell itself, there are only $\binom32=3$ many possibilities to select two symbols from the 3 remaining ones. In the worst case, this gives a factor of 3 for each of the 4 basic letters, which corresponds to the factor $3^n$ in \eqref{eq:space-consumption}.

Looking at a row with $n$ cells, we can associate to each cell one of the letters $\LA$, $\RA$ or one of the states $\DA'$ or $\UA'$, referring to downward or upward arrows that are not involved in (potential) cycles. Altogether, this gives us the following count for the number of configurations and hence for the space consumption of our DP algorithm;
for reasons of readability, we set $k=n^{[\,]}_{\scalebox{.5}{\UA}}$ and $t=k+n^{[\,]}_{\scalebox{.5}{\DA}}$.
\begin{equation}\label{eq:space-consumption}
3^n\sum_{t=0}^n\binom{n}{t}4^{n-t}\left(\sum_{k=0}^{t}\binom{t}{k}8^k\right)=3^n\sum_{t=0}^n\binom{n}{t}4^{n-t}9^t=3^n13^n=39^n.
\end{equation}
For the update from one sweep row to the next, we need to cycle through the (only) $4^n$ `new words', comparing them against the (at most) $39^n$ many configurations of the `old sweep row'. Namely, the type of a letter is determined from the box structure. 
Altogether, this shows the claim of the theorem.
Again, we could optimize the counting of the running time, because there is a trade-off between the types and the degree of freedom, but this would not change the overall result.

\end{proof}

\section{Discussion}

Games often come with some sort of didactical message.
In our case, it is not hard to see the Roma puzzle, being defined as a board game, to be generalizable to a game on graphs.
In other words, one could imagine this to be a gentle introduction into graph-theoretic concepts.
The game Generalized Roma that we propose is played on a weakly connected directed graph $G_{\mathcal{R}}=(V,E)$ with a special vertex $c_{\mathcal{R}}$ , the Roma-vertex, having out-degree zero. Moreover,
there is a set of  hints $H\subseteq E$.
The task is to delete edges (not from the hints~$H$) so that the resulting graph is a directed acyclic graph that is weakly connected and has maximum out-degree one. This could act as an introduction to graph-theoretic notions like  spanning trees or feedback arc sets.  The proof of \autoref{lem:all-paths-to-Rome} is also valid in this setting; this means again that all paths lead to Rome.
So far, the box condition has been neglected. This can be modeled as follows. We introduce  colors to the possible different orientations of each edge, i.e., we have a mapping $\chi_E:E\to C_E$ for the set of edge colors $C_E$; also, we assign colors to vertices by a mapping $\chi_V:V\to C_V$. Then, the box condition says that for each vertex color $c\in C_V$, the set of vertices $\chi_V ^{-1}(c)$ obeys that no two oriented edges $e_1,e_2$ that originate from vertices from $\chi_V ^{-1}(c)$ have the same color, i.e., $\chi_E(e_1)\neq\chi_E(e_2)$. This corresponds to having different arrows in the original Roma boxes.
This generalization would also allow for another specialization and hence to different game board designs. For instance, instead of taking quadratic cells (and hence four directions for the arrows), one could also think of triangular cells (hence, three directions for the arrows) or hexagonal cells (with six directions for the arrows); the box conditions would have to be adapted, too.

As an algorithmic challenge, it would be nice to further improve on the algorithm proposed in \autoref{thm:DPalgo}. It is an open challenge to design alternative algorithms that obtain running times within $\Oh^*(c^n)$ for solving $n\times n$-Roma puzzles, best without using exponential space.
Notice that our approach could be also interpreted in terms of pathwidth, considering the `board graph' that is a grid.
In this connection, we like to point to van der Zanden's essay on puzzles and treewidth~\cite{DBLP:conf/birthday/Zanden20}. 
There are other types of game problems where the same challenge is `on the board', for instance regarding finding a minimum dominating set of queens on an $n\times n$ chess board, see~\cite{Fer10a}.
A possible way out to meet this challenge could be to use the fact that planar graphs (e.g.,  grid graphs) with $n^2$ vertices not only have treewidth of $n$, but also \emph{treedepth} of~$n$. Moreover, there have been recent papers \cite{DBLP:conf/stacs/HegerfeldK20,DBLP:conf/wg/NederlofPSW20} that show that single-exponential algorithms with parameter treedepth are possible for several problems, including those that involve connectivity and cycle questions like \textsc{Connected Vertex Cover} and \textsc{Hamiltonian Cycle}, and which use polynomial space only.
We also refer to the general discussions in~\cite{CheRRV2018}.
In this direction, there would be also the challenge to test different algorithmic approaches on concrete Roma puzzles, including more standard approaches as using (I)LP solvers.


\nocite{doi:10.1137/S0097539793304601}
\nocite{DBLP:journals/tcs/DemaineMSWA18}
\nocite{LinWil2006}
 \clearpage
\bibliography{mybib}
\end{document}